\documentclass{article}

\usepackage{diagbox}
\usepackage{makecell}
\usepackage{booktabs}
\usepackage{tikz}
\usepackage{amsmath}
\usepackage{amsfonts}
\usepackage{amssymb}
\usepackage{amsthm}
\usepackage{url,ifthen}
\usepackage{enumitem}
\usepackage{srcltx}
\usepackage{multirow}
\usepackage{boxedminipage}
\usepackage[margin=1.1in]{geometry}
\usepackage{nicefrac}
\usepackage{xspace}
\usepackage{graphicx}
\usepackage{color}
\usepackage{subfigure}
\usepackage{colortbl}
\usepackage{setspace}
\usepackage{algorithm,algorithmic}
\usepackage[algo2e]{algorithm2e}
\usepackage[numbers]{natbib}

\SetKwRepeat{Do}{do}{until}

\definecolor{DarkGreen}{rgb}{0.1,0.5,0.1}
\definecolor{DarkRed}{rgb}{0.5,0.1,0.1}
\definecolor{DarkBlue}{rgb}{0.1,0.1,0.5}
\definecolor{Gray}{rgb}{0.2,0.2,0.2}

\usepackage{listings}
\lstdefinestyle{mystyle}{
    commentstyle=\color{DarkBlue},
    keywordstyle=\color{DarkRed},
    numberstyle=\tiny\color{Gray},
    stringstyle=\color{DarkGreen},
    basicstyle=\footnotesize,
    breakatwhitespace=false,
    breaklines=true,
    captionpos=b,
    keepspaces=true,
    numbers=left,
    numbersep=5pt,
    showspaces=false,
    showstringspaces=false,
    showtabs=false,
    tabsize=2
}
\usepackage[small]{caption}
\usepackage[pdftex]{hyperref}
\hypersetup{
    unicode=false,          
    pdftoolbar=true,        
    pdfmenubar=true,        
    pdffitwindow=false,      
    pdfnewwindow=true,      
    colorlinks=true,       
    linkcolor=DarkBlue,          
    citecolor=DarkGreen,        
    filecolor=DarkRed,      
    urlcolor=DarkBlue,          
    pdftitle={},
    pdfauthor={},
}

\def\draft{0}

\def\submit{0}

\ifnum\draft=1 
    
\else
    
\fi

\ifnum\submit=1
\newcommand{\forsubmit}[1]{#1}
\newcommand{\forreals}[1]{}
\else
\newcommand{\forreals}[1]{#1}
\newcommand{\forsubmit}[1]{}
\fi

\newtheorem{theorem}{Theorem}[section]
\newtheorem{remark}[theorem]{Remark}
\newtheorem{lemma}[theorem]{Lemma}
\newtheorem{corollary}[theorem]{Corollary}
\newtheorem{proposition}[theorem]{Proposition}
\newtheorem{fact}[theorem]{Fact}

\theoremstyle{definition}
\newtheorem{definition}[theorem]{Definition}

\newtheorem{example}[theorem]{Example}

\newtheoremstyle{example_contd}
{\topsep} {\topsep}%
{}
{}
{\bfseries}
{.}
{1em}
{\thmname{#1} \thmnumber{ #2}\thmnote{#3} (continued)}

\theoremstyle{example_contd}

\newcommand{\chapterref}[1]{\hyperref[ch:#1]{Chapter~\ref{ch:#1}}}

\newcommand{\claimref}[1]{\hyperref[claim:#1]{Claim~\ref{claim:#1}}}

\newcommand{\corollaryref}[1]{\hyperref[cor:#1]{Corollary~\ref{cor:#1}}}

\newcommand{\definitionref}[1]{\hyperref[def:#1]{Definition~\ref{def:#1}}}

\newcommand{\equationref}[1]{\hyperref[eq:#1]{Equation~\ref{eq:#1}}}

\newcommand{\factref}[1]{\hyperref[fact:#1]{Fact~\ref{fact:#1}}}

\newcommand{\figureref}[1]{\hyperref[fig:#1]{Figure~\ref{fig:#1}}}

\newcommand{\tableref}[1]{\hyperref[tab:#1]{Table~\ref{tab:#1}}}

\newcommand{\itemref}[1]{\hyperref[item:#1]{Item~(\ref{item:#1})}}

\newcommand{\lemmaref}[1]{\hyperref[lem:#1]{Lemma~\ref{lem:#1}}}

\newcommand{\propref}[1]{\hyperref[prop:#1]{Proposition~\ref{prop:#1}}}

\newcommand{\propositionref}[1]{\hyperref[prop:#1]{Proposition~\ref{prop:#1}}}

\newcommand{\remarkref}[1]{\hyperref[rem:#1]{Remark~\ref{rem:#1}}}

\newcommand{\sectionref}[1]{\hyperref[sec:#1]{Section~\ref{sec:#1}}}

\newcommand{\theoremref}[1]{\hyperref[thm:#1]{Theorem~\ref{thm:#1}}}

\usepackage[utf8]{inputenc}
\usepackage[T1]{fontenc}
\usepackage{kpfonts}
\usepackage{microtype}

\newcommand{\Esymb}{\mathbb{E}}
\newcommand{\Psymb}{\mathbb{P}}

\DeclareMathOperator*{\E}{\Esymb}

\DeclareMathOperator*{\ProbOp}{\Psymb r}
\renewcommand{\Pr}{\ProbOp}

\newcommand{\Prob}[1]{\Pr\left[ #1 \right]}

\newcommand{\one}{\mathbf{1}}

\newcommand{\A}{{\cal A}}
\newcommand{\F}{{\cal F}}
\newcommand{\G}{{\cal G}}
\newcommand{\X}{{\cal X}}

\newcommand{\cS}{{\cal S}}

\newcommand{\defeq}{\stackrel{\small \mathrm{def}}{=}}
\renewcommand{\leq}{\leqslant}

\renewcommand{\geq}{\geqslant}

\newcommand{\R}{\mathbb{R}}
\newcommand{\I}{\mathbb{I}}

\newcommand{\N}{\mathbb N}

\usepackage{bm}

\newcommand{\ignore}[1]{}

\renewcommand{\epsilon}{\varepsilon}

\newcommand{\eps}{\epsilon}

\newcommand{\remove}[1]{}

\newcommand{\cont}{\mathrm{CONT}}
\newcommand{\halt}{\mathrm{HALT}}
\newcommand{\Trestart}{T_{\mathrm{restart}}}

\newcommand{\Afilt}{\A^{\text{filt}}}
\newcommand{\loss}{\mathrm{Loss}}
\newcommand{\lossfilt}{\mathrm{Loss}^{\text{filt}}}
\newcommand{\bnorm}{B_{\text{norm}}}

\title{Individual Privacy Accounting via a R\'enyi Filter}
\author{Vitaly Feldman \\ Apple \and Tijana Zrnic\thanks{Work done while at Apple.}\\ University of California, Berkeley}
\date{}

\begin{document}

\maketitle

\begin{abstract}
We consider a sequential setting in which a single dataset of individuals is used to perform adaptively-chosen analyses, while ensuring that the differential privacy loss of each participant does not exceed a pre-specified privacy budget. The standard approach to this problem relies on bounding a worst-case estimate of the privacy loss over all individuals and all possible values of their data, for every single analysis. Yet, in many scenarios this approach is overly conservative, especially for ``typical'' data points which incur little privacy loss by participation in most of the analyses. In this work, we give a method for tighter privacy loss accounting based on the value of a personalized privacy loss estimate for each individual in each analysis. To implement the accounting method we design a {\em filter} for R\'enyi differential privacy. A filter is a tool that ensures that the privacy parameter of a composed sequence of algorithms with {\em adaptively-chosen} privacy parameters does not exceed a pre-specified budget. Our filter is simpler and tighter than the known filter for $(\epsilon,\delta)$-differential privacy by \citet{rogers2016privacy}. We apply our results to the analysis of noisy gradient descent and show that personalized accounting can be practical, easy to implement, and can only make the privacy-utility tradeoff tighter.
\end{abstract}

\section{Introduction}
Understanding how privacy of an individual degrades as the number of analyses using their data grows is of paramount importance in privacy-preserving data analysis. This allows individuals to participate in multiple disjoint statistical analyses, all the while knowing that their privacy cannot be compromised by aggregating the resulting reports. Furthermore, this feature is crucial for privacy-preserving algorithm design---instead of having to reason about the privacy properties of a complex algorithm, it allows reasoning about the privacy of the subroutines that make up the final algorithm.

For differential privacy \citep{DMNS06}, this accounting of privacy losses is typically done using composition theorems. Importantly, given that statistical analyses often rely on the outputs of previous analyses, and that algorithmic subroutines feed into one another, the composition theorems need to be \emph{adaptive}, namely, allow the choice of which algorithm to run next to depend on the outputs of all previous computations. For example, in gradient descent, the computation of the gradient depends on the value of the current iterate, which itself is the output of the previous steps of the algorithm.

Given the central role that adaptive composition theorems play in differentially private data analysis, they have been investigated in numerous works (e.g.~\cite{dwork2010boosting,kairouz2017composition,dwork2016concentrated,murtagh2016complexity,mironov2017renyi,bun2016concentrated,rogers2016privacy,dong2019gaussian,sommer2019privacy}). While they differ in some aspects, they also share one limitation. Namely, all of these theorems reason about the worst-case privacy loss for each constituent algorithm in the composition. Here, ``worst-case'' refers to the worst choice of individual in the dataset and worst choice of value for their data. This pessimistic accounting implies that every algorithm is summarized via a single privacy parameter, shared among all participants in the analysis. 

In most scenarios, however, different individuals have different effects on each of the algorithms, as measured by differential privacy.  More precisely, the output of an analysis may have little to no dependence on the presence of some individuals. For example, if we wish to report the average income in a neighborhood, removing an individual whose income is close to the average has virtually no impact on the final report after noise addition. Similarly, when training a machine learning model via gradient descent, the norm of the gradient given by a data point is often much smaller than the maximum norm (typically determined by a clipping operation). As a result, in many cases no single individual is likely to have the worst-case effect on all the steps of the analysis. This means that accounting based on existing composition theorems may be unnecessarily conservative.

In this work, we present a tighter analysis of privacy loss composition by computing the associated divergences at an individual level. In particular, to achieve a pre-specified privacy budget, we keep track of a personalized estimate of the privacy loss divergence for each individual in the analyzed dataset, and ensure that the respective estimate is maintained under the budget for all individuals throughout the composition. We do so by applying each analysis only to the points that are estimated to have sufficient leftover privacy budget.

The rest of the paper is organized as follows. In the remainder of this section, we give an overview of our main results and discuss related work. In the next section, we introduce the preliminaries necessary to state our results. In Section \ref{composition}, we prove our main adaptive composition theorem for R\'enyi differential privacy. We build off this result in Section \ref{filter}, where we develop a R\'enyi privacy filter---an object for budgeting privacy loss---and apply it to individual privacy accounting. In Section \ref{applications} we present an application of our theory to differentially private optimization, as well as some experimental results.

\subsection{Overview of main results}
It is feasible to measure the worst-case effect of a specific data point on a given analysis in terms of any of the divergences used to define differential privacy. One can simply  replace the supremum over all datasets in the standard definition of (removal) differential privacy with the supremum over datasets that include that specific data point (see Definition~\ref{def:stronger-ind-privacy}). Indeed, such a definition is given by \citet{ebadi2015differential} and a related definition is given by \citet{wang2019per}. However, a meaningful application of adaptive composition with such a definition immediately runs into the following technical challenge. Standard adaptive composition theorems require that the privacy parameter of each step be fixed in advance. For individual privacy parameters, this approach requires using the worst-case value of the individual privacy loss over all the possible analyses at a given step. Individual privacy parameters tend to be much more sensitive to the analysis being performed than worst-case privacy losses, and thus using the worst-case value over all analyses is likely to negate the benefits of using individual privacy losses in the first place.

Thus the main technical challenge in analyzing composition of individual privacy losses is that they are themselves random variables that depend on the outputs of all previous computations. More specifically, if we denote by $a_1,\dots,a_{t-1}$ the output of the first $t-1$ adaptively composed algorithms $\A_1,\dots,\A_{t-1}$, then the individual privacy loss of any point incurred by applying algorithm $\A_t$ is a function of $a_1,\dots,a_{t-1}$. 
Therefore, to tackle the problem of composing individual privacy losses we need to understand composition with \emph{adaptively-chosen} privacy parameters in general. We refer to this kind of composition as {\em fully adaptive}.

The setting of fully adaptive privacy composition is rather subtle and even defining privacy in terms of the adaptively-chosen privacy parameters requires some care. This setting was first studied by \citet{rogers2016privacy}, who introduced the notion of a \emph{privacy filter}. Informally, a privacy filter is a stopping time rule that halts a computation based on the adaptive sequence of privacy parameters and ensures that a pre-specified privacy budget is not exceeded. Rogers et al. define a filter for approximate differential privacy that asymptotically behaves like the advanced composition theorem \cite{dwork2010boosting}, but is substantially more involved and loses a constant factor. Moreover, several of the tighter analyses of Gaussian noise addition require composition to be done in R\'enyi differential privacy \citep{abadi2016deep,mironov2017renyi}. Converting them to $(\eps,\delta)$-differential privacy would incur an additional $\sqrt{\log(1/\delta)}$ factor in the final bound.

Our main result can be seen as a privacy filter for R\'enyi differential privacy (RDP) which justifies stopping the analyses based on the sum of privacy parameters so far \emph{even} under fully adaptive composition.

\begin{theorem}
\label{theorem:main}
Fix any $B\geq 0, \alpha\geq 1$. Suppose that $\A_t$ is $(\alpha,\rho_t)$-R\'enyi differentially private, where $\rho_t$ is an arbitrary function of $a_1,\dots,a_{t-1}$. If $\sum_{t=1}^k \rho_t \leq B$ holds almost surely, then the adaptive composition of $\A_1,\dots,\A_k$ is $(\alpha,B)$-R\'enyi differentially private.
\end{theorem}
Note that, when all privacy parameters are fixed, Theorem \ref{theorem:main} recovers the usual composition result for RDP \cite{mironov2017renyi}. Our RDP filter immediately implies a simple filter for approximate differential privacy that is as tight as any version of the advanced composition theorem obtained via concentrated differential privacy \cite{bun2016concentrated} (see Theorem~\ref{theorem:dp-filter-general}). These R\'enyi-divergence-based composition analyses are known to improve upon the classical rate of \citet{dwork2010boosting} and, in particular, improve on the rate in~\citep{rogers2016privacy}.

We instantiate our general result for fully adaptive composition in the setting of individual privacy accounting. This allows us to define an \emph{individual privacy filter}, which, given a fixed privacy budget, adaptively drops points from the analysis once their \emph{personalized} privacy loss estimate exceeds the budget. Therefore, instead of keeping track of a single running privacy loss estimate for all individuals, we track a less conservative, personalized estimate for each individual in the dataset. Individual privacy filtering allows for better, adaptive utilization of data points for a given budget. It can also naturally be applied to privacy accounting in the local differential privacy model, whereby each user stops responding once their local implementation of the filter indicates that their personal privacy budget is exhausted.

Individual privacy parameters are particularly easy to compute for linear
queries, as well as their high-dimensional generalizations. We show that our technique gives an algorithm for answering a sequence of adaptively-chosen linear queries that are sparse across time, meaning that, for any user, the number of queries that are non-zero on that user's data is small. Such queries arise, for example, when a platform counts the number of users that participate in certain activities (the type of activity being adaptive to the data collected in the previous days) and users generally participate in a small number of activities. Formally, a special case of our result implies the following theorem.
\begin{theorem}
\label{theorem:adaptive_queries_main}
There exists an algorithm $\A$ that, given a dataset $S=(X_1,\ldots,X_n)\in \X^n$, sparsity parameter $s$ and privacy level $\kappa$, for any adaptively-chosen sequence of queries $q_1,\ldots,q_k$ of arbitrary length $k$, where $q_i \colon \X \to \{0,1\}$, provides a sequence of answers $a_1,\ldots,a_k$ such that: $(1)$ $\A$ is $(\alpha,\alpha\kappa)$-RDP for all $\alpha \geq 1$;
$(2)$ for all $t$ and any $\delta \in (0,1)$, the probability that $|a_t - \sum_{X_i\in S_t} q_t(X_i)|> \sqrt{s\log(1/\delta)/\kappa}$ is at most $\delta$, where $S_t = (X_i\in S : \sum_{j=1}^t q_j(X_i) \leq s )$.
\end{theorem}
We note that the provided answers are guaranteed to be accurate only as long as the queries are truly sparse, meaning $\sum_{j=1}^t q_j(X_i) \leq s$ for (almost) all $i \in [n]$. This follows because the queries are accurate on the set $S_t$, hence $S_t$ needs to be similar to $S$ for the queries to be accurate on $S$. The privacy guarantee, on the other hand, holds for \emph{any} sequence of queries of \emph{any} length $k$. We describe a more general version of this result in Section~\ref{sec:adaptive-linear-queries}. A natural application of our general theorem is the setting of high-dimensional linear queries generated by gradient descent. We apply our theory to the analysis of private gradient descent \cite{abadi2016deep}, and show---both theoretically and empirically---that individual accounting can be easy to implement and can only make the resulting privacy-utility tradeoff tighter. Independently, without any individual accounting, in our empirical evaluations we also observe that private batch gradient descent, when tuned appropriately, outperforms private stochastic gradient descent in terms of the privacy-utility tradeoff. While we make this observation only on MNIST, we believe this phenomenon holds more generally and is worth further investigation.

\subsection{Related work}
The main motivation behind our work is obtaining tighter privacy accounting methods through, broadly speaking, ``personalized'' accounting of privacy losses. Existing literature in differential privacy discusses several related notions \cite{ghosh2011selling, ebadi2015differential,wang2019per,cummings2020individual}, although typically with an incomparable objective. \citet{ghosh2011selling} discuss individual privacy in the context of selling privacy at auction and their definition does not depend on the value of the data point but only on its index in the dataset. \citet{cummings2020individual} rely on a similar privacy definition, investigate an associated definition of individual sensitivity, and demonstrate a general way to preprocess an arbitrary function of a dataset into a function that has the desired bounds on individual sensitivities.

\citet{ebadi2015differential} introduce personalized differential privacy in the context of private database queries and describe a system which drops points when their personalized privacy loss exceeds a budget. In their system personalized privacy losses result from record selection operations applied to the database. While this type of accounting is similar to ours in spirit, their work only considers basic and non-adaptive composition. The work of \citet{wang2019per} considers the privacy loss of a specific data point relative to a fixed dataset and provides techniques for evaluating this ``per-instance'' privacy loss for several statistical problems. \citet{wang2019per} also briefly discusses adaptive composition of per-instance differential privacy as a straightforward generalization of the usual advanced composition theorem \cite{dwork2010boosting}, but the per-instance privacy parameters are assumed to be \emph{fixed}. As discussed above, having fixed per-instance privacy parameters, while allowing adaptive composition, is likely to negate the benefits of personalized privacy estimates. The work of \citet{ligett2020bounded} tightens individuals' personalized privacy loss by taking into account subsets of analyses in which an individual does not participate. Our work naturally captures this setting while allowing full adaptivity. Moreover, they consider the usual worst-case privacy loss, rather than an individual one, and the analyses in which a user participates are determined in a data-independent way.

Our work can be seen as related to data-dependent approaches to analyses of privacy-preserving algorithms such as smooth sensitivity \cite{nissim2007smooth}, the propose-test-release framework \cite{dwork2009differential}, and {\em ex-post} privacy guarantees \cite{wu2019accuracy}. Our results are complementary in that we aim to capture the dependence of the output on the value of each individual's data point as opposed to the ``easiness'' of the entire dataset. Our approach also relies on composition to exploit the gains from individual privacy loss accounting.

Finally, adaptive composition of differentially private algorithms is a key tool for establishing statistical validity of an adaptively-chosen sequence of statistical analyses \cite{dwork2015preserving, dwork2015generalization, bassily2016algorithmic}. In this context, \citet{feldman2018calibrating} show that individual KL-divergence losses (or RDP losses for $\alpha=1$) compose adaptively and can be used to derive tighter generalization results. However, their results still require that the average of individual KL-divergences be upper bounded by a fixed worst-case value and the analysis appears to be limited to the $\alpha=1$ case.

\section{Preliminaries}
\label{preliminaries}
We will denote by $S=(X_1,\dots,X_n)$ the analyzed dataset, and by $S^{-i} \defeq (X_1,\dots,X_{i-1},X_{i+1},\dots,X_n)$ the analyzed dataset after removing point $X_i$. We will generally focus on algorithms that can take as input a dataset of arbitrary size. If, instead, the algorithm requires an input of fixed size, one can obtain the same results for algorithms  that replace $X_i$ with an arbitrary element $X^\star$ fixed in advance (for example 0).

We start by reviewing some preliminaries on differential privacy.

\begin{definition}[\citep{DMNS06,ODO06}]
\label{def:dp}
A randomized algorithm $\A$ is $(\epsilon,\delta)$-differentially private (DP) if for all datasets $S = (X_1,\dots,X_{n})$,
$$\Prob{\A(S) \in E}\leq e^\epsilon\Prob{\A(S^{-i}) \in E} + \delta,\text{ and } \Prob{\A(S^{-i})\in E}\leq e^\epsilon\Prob{\A(S)\in E} + \delta,$$
for all $i\in[n]$ and all measurable sets $E$.
\end{definition}

Our analysis will rely on R\'enyi differential privacy (RDP), a relaxation of DP based on R\'enyi divergences which often leads to tighter privacy bounds than analyzing DP directly. Formally, the R\'enyi divergence of order $\alpha\in(1,\infty)$ between two measures $\mu$ and $\nu$ such that $\mu\ll \nu$ is defined as:
$$D_\alpha(\mu \| \nu) = \frac{1}{\alpha-1} \log \int \left(\frac{d\mu}{d\nu}\right)^\alpha d\nu.$$
The R\'enyi divergence of order $\alpha=1$ is defined by continuity, and recovers the Kullback-Leibler (KL) divergence. Relying on a common abuse of notation, we will use $\A(\cdot)$ to refer to the output distribution of a randomized algorithm. Thus, $D_\alpha(\A(S)\|\A(S^{-i}))$ denotes the divergence between the output distribution of $\A$ on input $S$ and the output distribution of $\A$ on input $S^{-i}$. Similarly, we will use $a\sim \A(S)$ to denote $a$ being sampled randomly from the output distribution of
$\A$ on $S$.
We also use the following shorthand notation for the maximum of the two directions of R\'enyi divergence:
$$D_\alpha^\leftrightarrow(\mu\|\nu) \defeq \max\left\{D_\alpha(\mu\|\nu), D_\alpha(\nu\|\mu)\right\}.$$

\begin{definition}[\citep{mironov2017renyi}]
\label{def:rdp-def}
A randomized algorithm $\A$ is $(\alpha,\rho)$-R\'enyi differentially private (RDP) if for all datasets $S = (X_1,\dots,X_{n})$,
\begin{align*}
D_\alpha^\leftrightarrow\left(\A(S)\|\A(S^{-i})\right) \leq \rho,
\end{align*}
for all $i\in[n]$.
\end{definition}

A related notion that we will make use of is zero-concentrated differential privacy (zCDP).

\begin{definition}[\citep{bun2016concentrated}]
A randomized algorithm $\A$ satisfies $\kappa$-zero-concentrated differential privacy (zCDP) if it satisfies $(\alpha,\alpha\kappa)$-RDP for all $\alpha\geq 1$.
\end{definition}

R\'enyi differential privacy implies differential privacy; therefore, although our guarantees will be stated in terms of RDP, the conversion to DP is immediate.

\begin{fact}[\citep{mironov2017renyi}]
\label{fact:conversion}
If algorithm $\A$ is $(\alpha,\rho)$-RDP, then it is also $\left(\rho + \frac{\log(1/\delta)}{\alpha-1},\delta\right)$-DP, for any $\delta\in(0,1)$.
\end{fact}

One of the successes of differential privacy (and RDP as well) lies in its adaptive composition property. In Algorithm \ref{alg:generalcomp} we define adaptive composition, which is at the center of our analysis.

\begin{algorithm}[H]
\SetAlgoLined
\SetKwInOut{Input}{input}
\Input{dataset $S \in\X^n$, sequence of algorithms $\A_t, t=1,2,\dots,k$}
\For{$t=1,\dots,k$}
{
Compute $a_t = \A_t(a_1,\dots,a_{t-1},S)$
}
Return $\A^{(k)}(S) = (a_1,\dots,a_k)$
\caption{Adaptive composition $\A^{(k)}$}
\label{alg:generalcomp}
\end{algorithm}
If $\A_t(a_1,\dots,a_{t-1},\cdot)$ is $(\alpha,\rho_t)$-RDP for \emph{all} values of $a_1,\dots,a_{t-1}$, then the standard adaptive composition theorem for RDP says that $\A^{(k)}$ is $(\alpha,\sum_{t=1}^k \rho_t)$-RDP \cite{mironov2017renyi}. Note that, by definition, the parameters $\rho_1,\dots,\rho_k$ are independent of the specific reports $a_1,\dots,a_{k}$ obtained in the adaptive computation. In other words, they are fixed in advance.

\subsection{Individual privacy losses}
Our individual accounting relies on measuring the maximum possible effect of an individual data point on a dataset statistic in terms of R\'enyi divergence. This measure is equivalent to an RDP version of personalized differential privacy \cite{ebadi2015differential}. For convenience we will refer to it as individual R\'enyi differential privacy, or individual RDP for short. We note, however, that, by itself, a bound on this divergence does \emph{not} imply any formal privacy guarantee for an individual, since the individual RDP parameter depends on the sensitive value of the data point.

\begin{definition}[Individual RDP]
\label{def:stronger-ind-privacy}
Fix $n\in\N$ and a data point $X$. We say that a randomized algorithm $\A$ satisfies $(\alpha, \rho)$-\emph{individual R\'enyi differential privacy for $X$} if for all datasets $S=(X_1,\ldots,X_m)$ such that $m\leq n$ and $X_i=X$ for some $i$, it holds that
$$D^\leftrightarrow_\alpha\left(\A(S)\|\A(S^{-i})\right) \leq \rho.$$
\end{definition}

Therefore, to satisfy the standard definition of RDP, an algorithm needs to satisfy individual RDP for all data points $X$.

Our main focus will be on individual privacy losses as introduced in Definition \ref{def:stronger-ind-privacy}, however some of our results also hold under a weaker notion of individual privacy loss, which measures the effect of a data point on the output of a statistical analysis, relative to a \emph{fixed} dataset. This notion is an RDP version of per-instance differential privacy \cite{wang2019per}.

\begin{definition}[Individual RDP (per-instance)]
\label{def:general-ind-privacy}
Fix a dataset $S=(X_1,\dots,X_n)$. We say that a randomized algorithm $\A$ satisfies $(\alpha, \rho)$-\emph{individual R\'enyi differential privacy for $(S,X_i)$} if it holds that
$$D^\leftrightarrow_\alpha\left(\A(S)\|\A(S^{-i})\right) \leq \rho.$$
\end{definition}

We will note which results hold under Definition \ref{def:general-ind-privacy}, in addition to being valid under Definition \ref{def:stronger-ind-privacy}.

Before we turn to analyzing composition, we give a simple example of individual RDP computation. For simplicity, we focus on Gaussian noise addition. Similar computations can be carried out for other randomization mechanisms.

\begin{example}[Linear queries]
\label{ex:linear-query}
Let $S=(X_1,\dots,X_n)\in\X^n$. Suppose that $\A$ is a $d$-dimensional linear query with Gaussian noise addition, $\A(S) = \sum_{j\in [n]} q(X_j) + \xi$, for some $q\colon \X \to \R^d$ and $\xi\sim N(0,\sigma^2 \I_d)$.
Then, $\A$ satisfies
$$\left(\alpha,D^\leftrightarrow_\alpha\left(N\left(\sum_{j\in[n]} q(X_j),\sigma^2\I_d\right)~\Bigg\|~ N\left(\sum_{j\in[n], j\neq i} q( X_j),\sigma^2\I_d\right)\right)\right) = \left(\alpha,\frac{\alpha \|q(X_i)\|_2^2}{2\sigma^2}\right)$$
individual RDP for $X_i$. Note that in this case individual RDP (Definition \ref{def:stronger-ind-privacy}) and per-instance RDP (Definition \ref{def:general-ind-privacy}) have the same value.
\end{example}

The analysis above extends to arbitrary Lipschitz functions.
\begin{example}[Lipschitz analyses]
Suppose that $g:(\R^d)^n\rightarrow \R^{d'}$ is $L_i$-Lipschitz in coordinate $i$ (in $\ell_2$-norm). For  $q\colon \X \to \R^d$, let $\A(S) = g(q(X_1),\dots,q(X_n)) + \xi$, $\xi\sim N(0,\sigma^2 \I_{d'})$. Assume that for some $X^\star$, $q(X^\star)$ is the origin.
Then, by using $X^\star$ to replace a removed element (namely, $S^{-i} = (X_1,\dots,X_{i-1},X^\star,X_{i+1},\dots,X_n)$), we get that  $\A$ satisfies $\left(\alpha,\frac{\alpha L_i^2 \|q(X_i)\|_2^2}{2\sigma^2}\right)$-individual RDP for $X_i$.
\end{example}

\section{Fully adaptive composition for R\'enyi differential privacy}
\label{composition}

Our main technical contribution is a new adaptive composition theorem for R\'enyi differential privacy, which bounds the overall privacy loss in terms of the individual privacy losses of all data points. As argued earlier, the main challenge in understanding how individual privacy parameters compose is the fact that these parameters are random, rather than fixed. In what follows, we first state a general version of our main theorem, which bounds the privacy loss in adaptive composition in terms of a bound on the sequence of possibly random privacy parameters. Then, we instantiate this result in the context of individual privacy.

We set up some notation within the context of adaptive composition (Algorithm \ref{alg:generalcomp}). We denote by $a^{(t)}\defeq(a_1,\dots,a_t)$ the sequence of the first $t$ reports, and by $\A^{(t)}(\cdot)\defeq (\A_1(\cdot),\A_2(\A_1(\cdot),\cdot),\dots,\A_t(\A_1(\cdot),\dots,\cdot))$ the composed algorithm which produces $a^{(t)}$. For two datasets $S$ and $S'$, parameter $\alpha\geq 1$ and \emph{fixed} $a^{(t)}$, we let\footnote{All algorithms we will be considering in this paper, if not discrete, induce a density with respect to the Lebesgue measure. For such instances, replacing expressions such as $\Prob{\A^{(t)}(S)=a}$ with the density of $\A^{(t)}(S)$ at $a$ gives the analysis in the continuous case.}
$$\loss^{(t)}(a^{(t)};S,S',\alpha)\defeq\left(\frac{\Prob{\A^{(t)}(S) = a^{(t)}}}{\Prob{\A^{(t)}(S') = a^{(t)}}}\right)^\alpha.$$
Similarly, for fixed $a^{(t)}$ we also define
$$\loss_t(a^{(t)};S,S',\alpha)\defeq\left(\frac{\Prob{\A_t(a_1,\dots,a_{t-1},S) = a_t}}{\Prob{\A_t(a_1,\dots,a_{t-1},S') = a_t}}\right)^\alpha.$$
Roughly speaking, $\loss^{(t)}$ denotes the total privacy loss incurred by the first $t$ rounds of adaptive composition, while $\loss_t$ denotes the loss incurred in round $t$ (which, due to adaptivity, depends on the outcomes of the first $t-1$ rounds).

Generally, we will be interested in $\loss_t(a^{(t)};S,S',\alpha)$ and $\loss^{(t)}(a^{(t)};S,S',\alpha)$ when $a^{(t)}$ is output by adaptive composition; in such cases, these two quantities are random.

Note that, since
\begin{align*}
    \Prob{\A^{(t)}(S) = a^{(t)}} &= \Prob{\A^{(t-1)}(S) = a^{(t-1)}}\Prob{\A_t(\A_1(S),\dots,\A_{t-1}(\A_1(S),\dots),S) = a_t~|~\A^{(t-1)}(S) = a^{(t-1)}}\\
    &= \Prob{\A^{(t-1)}(S) = a^{(t-1)}}\Prob{\A_t(a_1,a_2,\dots,a_{t-1},S) = a_t},
\end{align*}
we have $\loss^{(t)}(a^{(t)};S,S',\alpha)=\loss^{(t-1)}(a^{(t-1)};S,S',\alpha)\cdot\loss_t(a^{(t)};S,S',\alpha)$.

We let $\rho_t$ denote the RDP parameter of order $\alpha$ of $\A_t$, conditional on the past reports. For the sake of generality and simplicity of exposition, we introduce an abstract space $\cS$ over pairs of datasets and let
\begin{align}
\label{eqn:generalrho}
    \rho_t &\defeq \sup_{(S,S')\in\cS} D_\alpha^\leftrightarrow\left(\A_t(a_1,\dots,a_{t-1}, S)\|\A_t(a_1,\dots,a_{t-1}, S')\right)\\
    &= \frac{1}{\alpha-1}\log~ \sup_{(S,S')\in\cS}~\max\left\{\E_{a^{(t)}\sim \A^{(t)}(S')} \left[\loss_t(a^{(t)};S,S',\alpha)~\Big|~a^{(t-1)}\right], \E_{a^{(t)}\sim \A^{(t)}(S)} \left[\loss_t(a^{(t)};S',S,\alpha)~\Big|~a^{(t-1)}\right]\right\}. \nonumber
\end{align}
In the context of individual privacy (Definition \ref{def:stronger-ind-privacy}), we will instantiate $\cS$ to be the space of all dataset pairs where either dataset is obtained by deleting $X_i$ from the other. For the per-instance notion (Definition \ref{def:general-ind-privacy}), we will set $\cS = \left\{(S,S^{-i})\right\}$. In the context of usual RDP, $\cS$ will be the space of all pairs of datasets that differ in the presence of one element.

The classical composition theorem for R\'enyi differential privacy---while allowing $\A_t$ to depend on the previous reports---constrains $\A_t$ to be $(\alpha,\rho_t)$-RDP for some \emph{fixed} $\rho_t$. Here, we make no such constraints on $\rho_t$; hence, $\rho_t$ will in general be a random variable, due to the randomness in $a_1,\dots,a_{t-1}$.

Theorem \ref{theorem:composition} states that, as long as $\sum_{t=1}^k \rho_t$ is maintained under a fixed budget, the output of adaptive composition preserves privacy.

\begin{theorem}
\label{theorem:composition}
Fix any $B\geq 0,\alpha \geq 1$, and a set of pairs of datasets $\cS$. For any sequence of algorithms $\A_1,\ldots,\A_k$, if $\sum_{t=1}^k \rho_t \leq B$ holds almost surely, where the sequence $\rho_1,\ldots,\rho_k$ is defined in eq.~\eqref{eqn:generalrho}, then the adaptive composition $\A^{(k)}$ given in Algorithm~\ref{alg:generalcomp} satisfies
$$D_\alpha^\leftrightarrow\left(\A^{(k)}(S)\|\A^{(k)}(S')\right) \leq B,$$
for all $(S,S')\in\cS$.
\end{theorem}

\begin{proof}
Fix any $(S,S')\in\cS$. In what follows, we take $a^{(t)} = (a_1,\dots,a_t)$ to be distributed as the random output of adaptive composition applied to $S'$, that is $a^{(t)}\sim \A^{(t)}(S')$. Consequently, $~\loss^{(t)}(a^{(t)};S,S',\alpha)$ and $\loss_t(a^{(t)};S,S',\alpha)$ are also random.

Let $M_t \defeq \loss^{(t)}(a^{(t)};S,S',\alpha) e^{-(\alpha-1)\sum_{j=1}^t \rho_j}$, and let $M_0=1$. Consider the filtration $\Gamma_t = \sigma(a^{(t)})$. We prove that $M_t$ is a supermartingale with respect to $\Gamma_t$; that is, we show $ \E[M_t ~|~\Gamma_{t-1}] \leq M_{t-1}$. This follows since:
\begin{align*}
    \E[M_t ~|~\Gamma_{t-1}] &= \E\left[\loss^{(t)}(a^{(t)};S,S',\alpha) ~e^{-(\alpha-1)\sum_{j=1}^t \rho_j} ~\Bigg|~\Gamma_{t-1}\right]\\
    &= \E\left[\loss^{(t-1)}(a^{(t-1)};S,S',\alpha) ~ \loss_t(a^{(t)};S,S',\alpha) ~e^{-(\alpha-1)\sum_{j=1}^t \rho_j} ~\Bigg|~\Gamma_{t-1}\right]\\
    &= \loss^{(t-1)}(a^{(t-1)};S,S',\alpha)~ e^{-(\alpha-1)\sum_{j=1}^t \rho_j}\E\left[\loss_t(a^{(t)};S,S',\alpha) ~\Bigg|~\Gamma_{t-1}\right]\\
    &\leq \loss^{(t-1)}(a^{(t-1)};S,S',\alpha)~ e^{-(\alpha-1)\sum_{j=1}^t \rho_j}e^{(\alpha-1)\rho_t}\\
    &= \loss^{(t-1)}(a^{(t-1)};S,S',\alpha)~ e^{-(\alpha-1)\sum_{j=1}^{t-1} \rho_j}\\
    &= M_{t-1},
\end{align*}
where the third equality uses the fact that $(\rho_j)_{j=1}^t\in\Gamma_{t-1}$, and the inequality applies the definition of $\rho_t$. Therefore, by applying iterated expectations, we can conclude
$$\E[M_k] = \E_{a^{(k)}\sim \A^{(k)}(S')}\left[\loss^{(k)}\left(a^{(k)};S,S',\alpha\right) e^{-(\alpha-1)\sum_{j=1}^k \rho_j}\right]\leq \E[M_0] = 1.$$
Since $\sum_{j=1}^k \rho_j\leq B$ by assumption, this inequality implies
$$\E_{a^{(k)}\sim \A^{(k)}(S')}\left[\loss^{(k)}\left(a^{(k)};S,S',\alpha\right)\right]\leq e^{(\alpha-1)B}.$$
After normalizing, we get
$$D_\alpha\left(\A^{(k)}(S),\A^{(k)}(S')\right) = \frac{1}{\alpha -1}\log~\E_{a^{(k)}\sim \A^{(k)}(S')}\left[\loss^{(k)}\left(a^{(k)};S,S',\alpha\right) \right]\leq B.$$
The same argument can be used to bound the other direction of the divergence. Since the choice of $(S,S')$ was arbitrary, we can conclude $\sup_{(S,S')\in\cS}D^\leftrightarrow_\alpha\left(\A^{(k)}(S),\A^{(k)}(S')\right)\leq B$, as desired.
\end{proof}

A related argument is presented by \citet{cesar2020unifying} (see Lemma 3.1), who analyze privacy composition when a pre-specified set of concentrated differential privacy (CDP) parameters is adaptively ordered.

\begin{remark}
We remark that Theorem \ref{theorem:composition} is satisfied for any set $\cS$, including the singleton $\cS = \{(S,S')\}$. Denote by $\rho_t(S,S')$ the privacy parameter as defined in equation \eqref{eqn:generalrho} when $\cS = \{(S,S')\}$. Then, Theorem \ref{theorem:composition} implies that the adaptive composition $\A^{(k)}$ satisfies $(\alpha,B)$-RDP if for all datasets $(S,S')$ differing in the presence of one individual, $\sum_{t=1}^k \rho_t(S,S') \leq B$. In other words, in principle it is possible to place the maximum over $(S,S')\in\cS$ from the definition \eqref{eqn:generalrho} in front of the sum over privacy parameters. However, while this is formally a less conservative privacy accounting method, it remains unclear if it can be efficiently implemented in practice.
\end{remark}

We now instantiate Theorem \ref{theorem:composition} in the context of individual privacy.

To simplify notation, for a fixed point $X$, we let $\cS(X,n)$ denote the set of all dataset pairs $(S,S')$ such that $|S|\leq n$ and $S'$ is obtained by deleting element $X$ from $S$. More precisely, $(S,S')\in\cS(X,n)$ if $S=(X_1,\dots,X_m)$, where $m\leq n$ and $X_{i} = X$ for some $i$, and $S' = S^{-i}$.

We use $\rho_t^{(i)}$ to denote the individual privacy parameter of the $t$-th adaptively composed algorithm $\A_t$ with respect to $X_i$, conditional on the past reports. Formally, for fixed $\alpha\geq 1$ and for any data point $X_i\in S$ we let:
\small
\begin{align}
\label{eqn:stronger-ind-rho}
    \rho_t^{(i)} &\defeq \sup_{(S,S')\in\cS(X_i,n)} D_\alpha^\leftrightarrow \left(\A_t(a_1,\dots,a_{t-1},S)\|\A_t(a_1,\dots,a_{t-1},S')\right)\\
    &= \frac{1}{\alpha-1}\log\sup_{(S,S')\in\cS(X_i,n)} \max\left\{\E_{a^{(t)}\sim \A^{(t)}(S')} \left[\loss_t(a^{(t)};S,S',\alpha)~\Big|~a^{(t-1)}\right], \E_{a^{(t)}\sim \A^{(t)}(S)} \left[\loss_t(a^{(t)};S',S,\alpha)~\Big|~a^{(t-1)}\right]\right\}.\nonumber
\end{align}
\normalsize
Since $\rho_t^{(i)}$ is an instance of the general definition \eqref{eqn:generalrho} obtained by specifying $\cS$, a direct corollary of Theorem \ref{theorem:composition} is as follows.

\begin{corollary}
\label{corollary:ind-composition}
Fix any $B\geq 0$. If for any input dataset $S=(X_1,\dots,X_n)$, $\sum_{t=1}^k \rho_t^{(i)} \leq B$ holds almost surely for all individuals $i\in[n]$, then the adaptive composition $\A^{(k)}$ given in Algorithm~\ref{alg:generalcomp} is $(\alpha,B)$-R\'enyi differentially private.
\end{corollary}

\begin{proof}
By Theorem \ref{theorem:composition}, $\sum_{t=1}^k \rho_t^{(i)} \leq B$ implies that
$$D^\leftrightarrow_\alpha\left(\A^{(k)}(S)\|\A^{(k)}(S^{-i})\right)\leq B.$$
Since this holds for all $S\in \X^n$ and $i\in[n]$, we conclude that $\A^{(k)}$ is $(\alpha,B)$-R\'enyi differentially private.
\end{proof}

Notice that $\sum_{t=1}^k \rho_t^{(i)} \leq B$ is a \emph{data-specific} requirement, while classical composition results consider \emph{all} hypothetical datasets. In Section \ref{ind-filter} we will show how Corollary \ref{corollary:ind-composition} can be operationalized.

It is worth mentioning that Corollary \ref{corollary:ind-composition} also holds under the per-instance notion of individual privacy (Definition \ref{def:general-ind-privacy}). This result is obtained by simply taking $\cS = \{(S,S^{-i})\}$ in the proof, where $S$ is the analyzed dataset. However, our main application of Corollary \ref{corollary:ind-composition}---individual privacy filtering, stated in the following section---requires individual privacy loss accounting according to Definition \ref{def:stronger-ind-privacy}. 

\section{R\'enyi privacy filter}
\label{filter}
Fully adaptive composition was first studied by \citet{rogers2016privacy}. They defined the notion of a \emph{privacy filter}, a function that takes as input adaptively-chosen DP parameters $\epsilon_1,\delta_1,\dots,\epsilon_t,\delta_t$, as well as a global differential privacy budget $\epsilon_g,\delta_g$, and outputs $\cont$ if the overall report after $t$ rounds of adaptive composition with the corresponding privacy parameters is guaranteed to satisfy $(\epsilon_g,\delta_g)$-DP. Otherwise, it outputs $\halt$.

We now show that Theorem \ref{theorem:composition} immediately implies a simple RDP analogue of a privacy filter. Specifically, we show that by simply adding up privacy parameters, as in the usual composition where all privacy parameters are fixed up front, we obtain a filter for RDP. Our individual privacy accounting method can naturally be seen as a privacy filter applied to each data point individually. We formalize this in Section \ref{ind-filter}. Then, in Section \ref{sec:eps-filter}, we show that existing conversions of RDP guarantees to DP guarantees imply a new filter for $(\eps,\delta)$-DP that both simplifies and improves on the results of \citet{rogers2016privacy}.

We now define an RDP filter formally. This general definition is used primarily to explain the relationship of our results to the notions and results in \citep{rogers2016privacy}. Our individual privacy filtering application can be derived from Theorem \ref{theorem:composition} directly.

As in equation \eqref{eqn:generalrho}, we let $\rho_t$ denote the possibly random RDP parameter of order $\alpha$ of $\A_t$, conditional on the past reports. We again assume an implicit space $\cS$ over pairs of datasets, which we instantiate in different ways depending on the privacy accounting method.

\begin{algorithm}[H]
\SetAlgoLined
\SetKwInOut{Input}{input}
\Input{dataset $S \in\X^n$, maximum number of rounds $N\in\N$, sequence of algorithms $\A_k, k=1,2,\dots,N$}
\hspace{0.15cm}Initialize $k=0$\newline
\While{$k< N$}
{
\hspace{0.15cm}Compute $\rho_{k+1} = \sup_{(S_1,S_2)\in\cS} D^\leftrightarrow_\alpha\left(\A_{k+1}(a_1,\dots,a_{k}, S_1)\|\A_{k+1}(a_1,\dots,a_{k}, S_2)\right)$  \newline
\If{$\F_{\alpha, B}\left(\rho_1,\dots,\rho_{k+1} \right)=\halt$}{BREAK}
\hspace{0.15cm}Compute $a_{k+1} = \A_{k+1}(a_1,\dots,a_{k}, S)$\newline
$k\leftarrow k +1$
}
Return $\A^{(k)}(S) = (a_1,\dots,a_{k})$
\caption{Adaptive composition with R\'enyi privacy filtering}
\label{alg:filteredcomp2}
\end{algorithm}

As a remark, although in Algorithm \ref{alg:filteredcomp2} we write ``compute $\rho_{k+1}$'', sometimes the exact computation is infeasible and we actually only require the computed quantity to be an upper bound on the exact divergence. We suppress this distinction for readability. The same remark applies to other algorithm displays in the paper.

Let $S_\infty$ denote the set of all positive, real-valued finite sequences.
\begin{definition}[RDP filter]
\label{def:rdp-filter}
Fix a parameter $\alpha\geq 1$, and privacy budget $B$. We say that $\F_{\alpha, B}:S_\infty \rightarrow \{\cont,\halt\}$ is a valid R\'enyi privacy filter, or RDP filter for short, if for any sequence of algorithms $(\A_k)_{k=1}^N$ and any pair of datasets $(S_1,S_2)$, $\A^{(k)}$ given in Algorithm \ref{alg:filteredcomp2} with $\cS = \{(S_1,S_2)\}$, satisfies
$$D^\leftrightarrow_\alpha\left(\A^{(k)}(S_1)\|\A^{(k)}(S_2)\right)\leq B.$$
\end{definition}

We note that we can assume, without loss of generality, that the filter is monotone, namely that if $\F_{\alpha, B}\left(\rho_1,\dots,\rho_k \right)=\cont$ then $\F_{\alpha, B}\left(\rho'_1,\dots,\rho'_k \right)=\cont$ whenever $\rho'_i\leq \rho_i$ for all $i\in[k]$. This immediately implies that an RDP filter can be applied with $\rho_k$ defined using an arbitrary set of dataset pairs $\cS$ instead of just a single pair of datasets.

\begin{lemma}
Fix a parameter $\alpha\geq 1$, and privacy budget $B$. If $\F_{\alpha, B}:S_\infty \rightarrow \{\cont,\halt\}$ is a valid RDP filter, then for any sequence of algorithms $(\A_k)_{k=1}^N$ and any set of pairs of datasets $\cS$, $\A^{(k)}$ given in Algorithm \ref{alg:filteredcomp2} satisfies
$$\sup_{(S_1,S_2)\in\cS} D^\leftrightarrow_\alpha\left(\A^{(k)}(S_1)\|\A^{(k)}(S_2)\right)\leq B.$$
\end{lemma}

We remark that the analyst might choose an algorithm at time $t$ that exceeds the privacy budget, which will trigger the filter $\F_{\alpha,B}$ to halt. However, the analyst can then decide to change the computation at time $t$ retroactively and query the filter again, which then might allow continuation. This way, one can ensure a sequence of $N$ computations with formal privacy guarantees, for any target number of rounds $N$. In the following subsection, we present an application of RDP filters to individual privacy loss accounting which relies on this reasoning.

\begin{theorem}
\label{theorem:rdp-at-stopping-times}
Let
$$\F_{\alpha, B}(\rho_1,\dots,\rho_k) = \begin{cases}
\cont, \text{ if } \sum_{t=1}^k \rho_t \leq B, \\
\halt, \text{ if } \sum_{t=1}^k \rho_t > B.
\end{cases}$$
Then, $\F_{\alpha, B}$ is a valid R\'enyi privacy filter.
\end{theorem}

\begin{proof}
The only difference between Theorem \ref{theorem:composition} and this theorem is that a privacy filter halts at a random round, meaning the length of the output is random rather than fixed. Therefore, in this proof we formalize the fact that Theorem \ref{theorem:composition} is valid even under adaptive stopping.

Fix any $(S_1,S_2)$. By the argument in Theorem \ref{theorem:composition}, $M_t \defeq \loss^{(t)}(a^{(t)};S_1,S_2,\alpha) e^{-(\alpha-1)\sum_{j=1}^t \rho_j}$, where $a^{(t)}\sim \A^{(t)}(S_2)$, is a supermartingale with respect to $\Gamma_t = \sigma(a^{(t)})$.

Let $T$ be the last round before Algorithm \ref{alg:filteredcomp2} halts; that is,
$$T=\min\left\{t:\F_{\alpha,B}(\rho_1,\dots,\rho_{t+1}) = \halt\right\}\wedge N.$$

Note that $T$ is a stopping time with respect to $\Gamma_t$, that is $\{T=t\}\in\Gamma_t$, due to the fact that $\rho_{t+1}\in\Gamma_t$. Since $T$ is almost surely bounded by construction, we can apply the optional stopping theorem for supermartingales to get
$$\E[M_T] = \E_{a^{(T)}\sim \A^{(T)}(S_2)}\left[\loss^{(T)}\left(a^{(T)};S_1,S_2,\alpha\right) e^{-(\alpha-1)\sum_{j=1}^T \rho_j}\right] \leq \E[M_0] = 1.$$
By definition of the RDP filter, we know that $\sum_{j=1}^T \rho_j \leq B$ almost surely; otherwise the filter would have halted earlier. Thus, we can conclude
$$\E_{a^{(T)}\sim \A^{(T)}(S_2)}\left[\loss^{(T)}\left(a^{(T)};S_1,S_2,\alpha\right) e^{-(\alpha-1)B} \right]\leq 1.$$
After rearranging and normalizing, this implies
$$D_\alpha\left(\A^{(T)}(S_1),\A^{(T)}(S_2)\right) = \frac{1}{\alpha -1}\log~\E_{a^{(T)}\sim \A^{(T)}(S_2)}\left[\loss^{(T)}\left(a^{(T)};S_1,S_2,\alpha\right) \right]\leq B.$$
The same argument can be used to bound the other direction of the divergence. Therefore, $\F_{\alpha,B}$ is a valid RDP filter.
\end{proof}

\begin{remark}
For algorithms that satisfy zCDP, the stopping rule of Theorem \ref{theorem:rdp-at-stopping-times} suffices for controlling the overall zCDP privacy loss as well. Namely, if  $\A_t(a_1,\ldots,a_{t-1},\cdot)$ is $\kappa_t$-zCDP, then the halting criterion of the R\'enyi privacy filter with parameters $(\alpha,\alpha \kappa)$ is $\sum_{t=1}^k \alpha \kappa_t \leq \alpha \kappa$, which simplifies to $\sum_{t=1}^k \kappa_t \leq \kappa$. Since this stopping rule is independent of $\alpha$, we conclude that the overall output is $(\alpha,\alpha \kappa)$-RDP for all $\alpha\geq 1$, which is equivalent to $\kappa$-zCDP. More generally, Theorem \ref{theorem:rdp-at-stopping-times} extends to tracking R\'enyi privacy loss for any set of orders $\{\alpha_j\}_{j\in \mathcal{I}}$: if all filters in the set $\{\F_{\alpha_j,B}\}_{j\in\mathcal{I}}$ output $\cont$, then the output of adaptive composition satisfies $(\alpha_j,B)$-RDP, $\forall j\in\mathcal{I}$.
\end{remark}

We also remark that Rogers et al. define a privacy filter somewhat more generally, by treating the analyst as an adversary who is allowed to pick ``bad'' neighboring datasets at every step (see Algorithm 2 in \cite{rogers2016privacy}). Theorem \ref{theorem:rdp-at-stopping-times} holds under this setting as well, however we opted for a simpler presentation.

Just like Corollary \ref{corollary:ind-composition} applies Theorem \ref{theorem:composition} in the individual privacy setting, we can apply R\'enyi privacy filters to individual RDP parameters, in which case the filter indicates whether the privacy loss of a specific individual is potentially violated.

\subsection{Individual privacy accounting via a privacy filter}
\label{ind-filter}
Now we design an \emph{individual privacy filter}, which monitors individual privacy loss estimates across all individuals and all computations, and ensures that the privacy of all individuals is preserved. The filter guarantees privacy by adaptively dropping data points once their cumulative individual privacy loss estimate is about to cross a pre-specified budget. More specifically, at every step of adaptive composition $t$, it determines an active set of points $S_t\subseteq S$ based on cumulative estimated individual losses, and applies $\A_t$ only to $S_t$.



\begin{algorithm}[H]
\SetAlgoLined
\SetKwInOut{Input}{input}
\Input{dataset $S \in\X^n$, sequence of algorithms $\A_t$, $t=1,2,\dots,k$}
\For{$t = 1,\dots,k$}{
\hspace{0.1cm} For all $X_i\in S$, compute $\rho_t^{(i)} = \sup_{(S_1,S_2)\in\cS(X_i,n)} D_\alpha^\leftrightarrow \left(\A_t(a_1,\dots,a_{t-1},S_1)\|\A_t(a_1,\dots,a_{t-1},S_2)\right)$\newline
Determine active set $S_t = \left(X_i~:~\F_{\alpha,B}(\rho_1^{(i)},\dots,\rho_t^{(i)}) = \cont\right)$\newline
For all $X_i\in S$, set $\rho_t^{(i)} \leftarrow \rho_t^{(i)}\one\{X_i \in S_t\}$\newline
Compute $a_t = \A_t(a_1,\dots,a_{t-1},S_t)$
}
Return $(a_1,\dots,a_k)$
\caption{Adaptive composition with individual privacy filtering}
\label{alg:filteredcomp}
\end{algorithm}

Here, $\F_{\alpha,B}$ is the R\'enyi privacy filter from Theorem \ref{theorem:rdp-at-stopping-times}. Given its validity, one can observe that Algorithm \ref{alg:filteredcomp} preserves R\'enyi differential privacy.

\begin{theorem}
\label{theorem:ind-filter-via-rdp-filter}
Adaptive composition with individual privacy filtering (Algorithm \ref{alg:filteredcomp}) satisfies $(\alpha,B)$-R\'enyi differential privacy.
\end{theorem}

\begin{proof}
Denote by $\Afilt_t$ the subroutine given by the $t$-th step of the individual filtering algorithm; that is, $a_t = \Afilt_t(a_1,\dots,a_{t-1},S)$. Note that $\Afilt_t$ is \emph{not} equal to $\A_t$. By analogy with the notation $\A^{(t)}$, we also let $\A^{\text{filt}(t)}(\cdot)\defeq (\Afilt_1(\cdot),\Afilt_2(\Afilt_1(\cdot),\cdot),\dots,\Afilt_t(\Afilt_1(\cdot),\dots,\cdot))$.

We argue that the privacy loss of point $X_i$ in round $t$, conditional on the past reports, is upper bounded by $\rho_t^{(i)}$ (after $\rho_t^{(i)}$ has been updated):
\begin{equation*}
\label{eqn:condition}
    \frac{1}{\alpha-1}\log \max\left\{\E_{a^{(t)}\sim \A^{\text{filt}(t)}(S^{-i})} \left[\lossfilt_t(a^{(t)};S,S^{-i},\alpha)~\Big|~a^{(t-1)}\right], \E_{a^{(t)}\sim \A^{\text{filt}(t)}(S)} \left[\lossfilt_t(a^{(t)};S^{-i},S,\alpha)~\Big|~a^{(t-1)}\right]\right\} \leq \rho_t^{(i)},
\end{equation*}
where $\lossfilt_t(a^{(t)};S,S',\alpha) = \left(\frac{\Prob{\Afilt_t(a_1,\dots,a_{t-1},S) = a_t}}{\Prob{\Afilt_t(a_1,\dots,a_{t-1},S') = a_t}}\right)^\alpha.$

To do so, we reason about the active set of points at time $t$ when the input to adaptive composition is $S$, and when the input is $S^{-i}$. Denote by $S_t$ the active set given input $S$, and by $S_t^{(i)}$ the active set given input $S^{-i}$. Observe that, conditional on $a_1,\dots,a_{t-1}$, we have $(S_t,S_t^{(i)})\in \cS(X_i,n)$; that is, $S_t$ and $S_t^{(i)}$ differ only in the presence on $X_i$. This follows because the sequence $(\rho_j^{(i)})_{j=1}^t$ is measurable with respect to $a_1,\dots,a_{t-1}$, and whether point $X_i$ is active at time $t$ is in turn determined based only on $(\rho_j^{(i)})_{j=1}^t$. In particular, whether any given point is active does not depend on the rest of the input dataset ($S$ or $S^{-i}$), given $a_1,\dots,a_{t-1}$. Therefore, if $X_i\not\in S_t$, then $X_i$ loses no privacy in round $t$, because $\Afilt_t(a_1,\dots,a_{t-1}, S) \stackrel{d}{=} \Afilt_t(a_1,\dots,a_{t-1}, S^{-i})$, conditional on $a_1,\dots,a_{t-1}$. On the other hand, if $X_i\in S_t$, then its privacy loss can be bounded as
\begin{align*}
    &\frac{1}{\alpha-1}\log \max\left\{\E_{a^{(t)}\sim \A^{\text{filt}(t)}(S^{-i})} \left[\lossfilt_t(a^{(t)};S,S^{-i},\alpha)~\Big|~a^{(t-1)}\right], \E_{a^{(t)}\sim \A^{\text{filt}(t)}(S)} \left[\lossfilt_t(a^{(t)};S^{-i},S,\alpha)~\Big|~a^{(t-1)}\right]\right\}\\
    \leq~ &\frac{1}{\alpha-1}\log \max\left\{\E_{a^{(t)}\sim \A^{\text{filt}(t)}(S^{-i})} \left[\loss_t(a^{(t)};S_t,S_t^{(i)},\alpha)~\Big|~a^{(t-1)}\right], \E_{a^{(t)}\sim \A^{\text{filt}(t)}(S)} \left[\loss_t(a^{(t)};S_t^{(i)},S_t,\alpha)~\Big|~a^{(t-1)}\right]\right\}\\
    \leq~ &\frac{1}{\alpha-1}\log\sup_{(S_1,S_2)\in \cS(X_i,n)} \max\left\{\E_{a^{(t)}\sim \A^{(t)}(S_2)} \left[\loss_t(a^{(t)};S_1,S_2,\alpha)~\Big|~a^{(t-1)}\right], \E_{a^{(t)}\sim \A^{(t)}(S_1)} \left[\loss_t(a^{(t)};S_2,S_1,\alpha)~\Big|~a^{(t-1)}\right]\right\}\\
\leq~ &\rho_t^{(i)}.
\end{align*}
With this, we have showed that $\rho_t^{(i)}$ is a valid estimate of the privacy loss of $X_i$, for all $i\in[n]$.

Now we argue that, at the end of every round $t$ (after $\rho_t^{(i)}$ has been updated), $\F_{\alpha,B}(\rho_1^{(i)},\dots,\rho_t^{(i)}) = \cont$ for all $i\in[n]$. This follows by induction. For $t=1$, this is clearly true because $\F_{\alpha,B}(0) = \cont$. Now assume it is true at time $t-1$. Then, at time $t$, the filter clearly continues for all $X_i\in S_t$ simply by definition of $S_t$. If $X_i\not\in S_t$, then $\F_{\alpha,B}(\rho_1^{(i)},\dots,\rho_t^{(i)}) = \F_{\alpha,B}(\rho_1^{(i)},\dots,\rho_{t-1}^{(i)},0)= \F_{\alpha,B}(\rho_1^{(i)},\dots,\rho_{t-1}^{(i)})=\cont$. Therefore, we conclude that at the end of every round $t\in[k]$ and for all $i\in[n]$, the filter would output $\cont$. By the validity of $\F_{\alpha,B}$, we know that $\F_{\alpha,B}(\rho_1^{(i)},\dots,\rho_{t}^{(i)}) = \cont$ implies
$$D^\leftrightarrow_\alpha\left(\A^{(t)}(S)\|\A^{(t)}(S^{-i})\right)\leq B,$$
and since this holds for all $S$ and all $i\in[n]$, we conclude that Algorithm \ref{alg:filteredcomp} is $(\alpha,B)$-RDP.
\end{proof}

We can now justify the use of a supremum over all datasets that include $X$ in Definition \ref{def:stronger-ind-privacy} and, in particular, why the per-instance notion in Definition \ref{def:general-ind-privacy} does not suffice. By the current design of Algorithm \ref{alg:filteredcomp}, $X_i \not\in S_t$ implies no privacy loss for point $X_i$. This is true because, conditional on $a_1,\dots,a_{t-1}$, $X_i$ being inactive ensures that $S_t$ would be the same regardless of whether the input to Algorithm \ref{alg:filteredcomp} is $S$ or $S^{-i}$. Consequently, the output at time $t$ would be insensitive to the value of $X_i$. Under the more fine-grained definition of individual privacy, \emph{even if} $X_i \not\in S_t$, its privacy could still leak at round $t$. The reason is that, under two different inputs $S$ and $S^{-i}$, the running privacy loss estimates for all points are different, and hence the active set $S_t$ in the two hypothetical scenarios could be different as well. This fact, in turn, implies two different distributions over reports $a_t$. In short, dropping $X_i$ from the analysis \emph{does not} prevent its further privacy leakage if accounting is done according to Definition \ref{def:general-ind-privacy}.

\begin{remark}
Algorithm \ref{alg:filteredcomp} can naturally be applied to privacy accounting in the local differential privacy model. Here, each user would have a local implementation of the individual privacy filter and would stop responding when the filter halts. This is possible because the decision to halt for any given data point does not depend on the other data points other than through the sequence of reports $a_1,\dots,a_t$.
\end{remark}

In Section \ref{applications}, we apply the individual privacy filter to differentially private optimization via gradient descent, and demonstrate how this object ensures utilization of data points as long as their \emph{realized} gradients have low norm.

\subsection{Answering linear queries}
\label{sec:adaptive-linear-queries}
To illustrate the gains of individual privacy, we consider the task of answering adaptively-chosen high-dimensional linear queries. We aim to design an algorithm that receives a sequence of queries $q_1,q_2,\dots$, where $q_t:\X\rightarrow \R^d$ for all $t\in\N$, and upon receiving $q_t$ provides an estimate $a_t$ of $\sum_{i=1}^n q_t(X_i)$, where  $S=(X_1,\dots,X_n)\in \X^n$. In some applications it is natural to expect that, for a typical user, many of the queries evaluate to a very small value (having norm close to zero). For example, in the context of continual monitoring \cite{dwork2010differential, erlingsson2019amplification}, a platform might collect one real-valued indicator per user per day, and wish to make decisions based off the daily averages of these indicators across users. Here, $X_i$ would be a single user, and $q_t(X_i) = X_{t}^{(i)}$ would be the corresponding user's indicator on day $t$. For example, $X_t^{(i)}\in\{0,1\}$ could be a binary indicator of a change of some state for user $i$ on day $t$. For simplicity, we will treat the dataset $S$ as fixed, but our results apply to a more general setting in which the users' data can be updated after each query; for example, additional points might arrive in the process of the analysis (see, e.g., \cite{cummings2018differential}).



A prototypical mechanism for answering linear queries is the Gaussian mechanism, which reports $a_t=\sum_{i=1}^n q_t(X_i) + \xi_t$, where $\xi_t\sim N(0,\sigma^2\I_d)$. If the range of $q_t$ is constrained (or clipped) to have norm at most $C$, then the worst-case RDP loss incurred by answering $q_t$ is $(\alpha,\rho_t) = (\alpha,  \alpha C^2/(2\sigma^2))$. 
This implies that the standard analysis---which only considers $\rho_t$---allows answering at most $k_0= \lfloor 2B\sigma^2/(C^2\alpha)\rfloor$ queries in order to ensure $(\alpha,B)$-RDP. As we mentioned in Example~\ref{ex:linear-query}, the Gaussian mechanism satisfies $(\alpha,\rho_t^{(i)})$-individual RDP for $X_i$, where $\rho_t^{(i)} = \alpha \|q_t(X_i)\|_2^2/(2\sigma^2)$. Thus the individual privacy filter allows us to provide accurate answers to $q_t$ as long as each user's responses are ``sparse'' (more generally, have small $\sum_{j=1}^t \|q_j(X_i)\|_2^2$). Formally, we obtain the following generalization of Theorem~\ref{theorem:adaptive_queries_main}.

\begin{corollary}
\label{corollary:adaptive_queries}
There exists an algorithm $\A$ that, given a norm budget $\bnorm$ and privacy level $\kappa$, for any adaptively-chosen sequence of queries $q_1,\ldots,q_k$ of arbitrary length $k$, where $q_i \colon \X \to \R^d$, provides a sequence of answers $a_1,\ldots,a_k$ such that: $(1)$ $\A$ is $\kappa$-zCDP, that is, $(\alpha,\alpha\kappa)$-RDP for all $\alpha \geq 1$; $(2)$ for all $t$ and any $\delta \in (0,1)$, the probability that $\|a_t - \sum_{X_i\in S_t} q_t(X_i)\|_\infty > \sqrt{\bnorm \log(d/\delta)/\kappa}$ is at most $\delta$, where $S_t = (X_i\in S : \sum_{j=1}^t \|q_j(X_i)\|_2^2 \leq \bnorm )$.
\end{corollary}
Corollary \ref{corollary:adaptive_queries} follows from Theorem \ref{theorem:ind-filter-via-rdp-filter}, by setting each $\A_t$ to be the Gaussian mechanism with $\sigma^2 = \bnorm/(2\kappa)$. Note that, due to $\rho_t^{(i)}\leq \rho_t$, \emph{all} points are active in $S_t$ for at least the first $k_0$ computations, as prescribed by the usual worst-case analysis, and during those $k_0$ steps the answers are guaranteed to be accurate. Therefore, individual privacy provides a more fine-grained way of quantifying privacy loss by taking into account the value of the point whose loss we aim to measure. While $k$ is technically allowed to be arbitrarily large, after a certain number of reports we expect few points to remain active; we discuss stopping criteria in the following section.

\subsection{\texorpdfstring{$(\epsilon,\delta)$}{epsdelta}-differential privacy filter via R\'enyi filter}
\label{sec:eps-filter}

By connections between R\'enyi differential privacy and approximate differential privacy \cite{bun2016concentrated, mironov2017renyi}, we can translate our R\'enyi privacy filter into a filter for approximate differential privacy.

We define a valid DP filter analogously to a valid RDP filter, the difference being that it takes as input DP, rather than RDP parameters, and that it is parameterized by a global DP budget $\epsilon_g\geq 0,\delta_g\in(0,1)$. We denote by $\epsilon_t$ the possibly adaptive differential privacy parameter of $\A_t$:
\begin{align*}
    \epsilon_t &\defeq \sup_{(S_1,S_2)\in\cS} D^\leftrightarrow_\infty\left(\A_t(a_1,\dots,a_{t-1},S_1)\|\A_t(a_1,\dots,a_{t-1},S_2)\right)\\
    &= \sup_{(S_1,S_2)\in\cS} \max\left\{ D_\infty\left(\A_t(a_1,\dots,a_{t-1},S_1)\|\A_t(a_1,\dots,a_{t-1},S_2)\right), D_\infty\left(\A_t(a_1,\dots,a_{t-1},S_2)\|\A_t(a_1,\dots,a_{t-1},S_1)\right)\right\}.
\end{align*}
Here, $D_\infty$ denotes the max-divergence, obtained as the limit of R\'enyi divergence of order $\alpha$ by taking $\alpha\rightarrow\infty$.

We focus on advanced composition of \emph{pure} differentially private algorithms $\A_t$. As shown in \cite{rogers2016privacy}, a DP filter for approximately differentially private algorithms $\A_t$ can be obtained by an immediate extension of a filter for pure DP algorithms.

\begin{algorithm}[H]
\SetAlgoLined
\SetKwInOut{Input}{input}
\Input{dataset $S \in\X^n$, maximum number of rounds $N\in\N$, sequence of algorithms $\A_k, k=1,2,\dots,N$}
\hspace{0.15cm}Initialize $k=0$\newline
\While{$k< N$}
{
\hspace{0.15cm}Compute $\epsilon_{k+1} = \sup_{(S_1,S_2)\in\cS} D_\infty^\leftrightarrow\left(\A_{k+1}(a_1,\dots,a_{k}, S_1)\|\A_{k+1}(a_1,\dots,a_{k}, S_2)\right)$  \newline
\If{$\G_{\epsilon_g, \delta_g}\left(\epsilon_1,\dots,\epsilon_{k+1} \right)=\halt$}{BREAK}
\hspace{0.15cm}Compute $a_{k+1} = \A_{k+1}(a_1,\dots,a_{k}, S)$\newline
$k\leftarrow k +1$
}
Return $\A^{(k)}(S) = (a_1,\dots,a_{k})$
\caption{Adaptive composition with differential privacy filtering}
\label{alg:filteredcomp-dp}
\end{algorithm}

As before, $S_\infty$ denotes the set of all positive, real-valued finite sequences.

\begin{definition}[DP filter]
\label{def:dp-filter}
Fix $\epsilon_g\geq 0,\delta_g\in(0,1)$. We say that $\G_{\epsilon_g, \delta_g}:S_\infty \rightarrow \{\cont,\halt\}$ is a valid differential privacy filter, or DP filter for short, if for any sequence of algorithms $(\A_k)_{k=1}^N$ and any pair of datasets $(S_1,S_2)$, $\A^{(k)}$ given in Algorithm \ref{alg:filteredcomp-dp} with $\cS = \{(S_1,S_2)\}$, satisfies
$$\Prob{\A^{(k)}(S_1)\in E}\leq e^{\epsilon_g} \Prob{\A^{(k)}(S_2)\in E} +\delta_g\text{ and }\Prob{\A^{(k)}(S_2)\in E}\leq e^{\epsilon_g} \Prob{\A^{(k)}(S_1)\in E} +\delta_g,$$
for all measurable events $E$.
\end{definition}

We note that \citet{rogers2016privacy} define a differential privacy filter somewhat differently---for example, in their definition a filter admits a sequence of privacy parameters of fixed length---however Definition \ref{def:dp-filter} is essentially equivalent to theirs.

By invoking standard conversions between DP and R\'enyi-divergence-based privacy notions, our analysis implies a simple stopping condition for a DP filter, in terms of any zero-concentrated differential privacy (zCDP) level which ensures $(\epsilon_g,\delta_g)$-DP. For clarity, we give one particularly simple such translation from zCDP to DP, however one could in principle invoke more sophisticated analyses such as those of \citet{bun2016concentrated}. In general, we can reproduce any rate for advanced composition of DP that utilizes R\'enyi-divergence-based privacy definitions, in the setting of fully adaptive composition.

\begin{theorem}
\label{theorem:dp-filter-general}
Let $B^\star$ be the largest $B>0$ such that $B$-zero-concentrated differential privacy (zCDP) implies $(\epsilon_g,\delta_g)$-differential privacy. Let
$$\G_{\epsilon_g,\delta_g}(\epsilon_1,\dots,\epsilon_k) = \begin{cases}
\cont, \text{ if } \frac{1}{2}\sum_{t=1}^k \epsilon_t^2 \leq B^\star, \\
\halt, \text{ if } \frac{1}{2}\sum_{t=1}^k \epsilon_t^2 > B^\star.
\end{cases}$$
Then, $\G_{\epsilon_g,\delta_g}$ is a valid DP filter. For example,
$$\G_{\epsilon_g,\delta_g}(\epsilon_1,\dots,\epsilon_k) = \begin{cases}
\cont, \text{ if } \frac{1}{2}\sum_{t=1}^k \epsilon_t^2 \leq \left(-\sqrt{\log(1/\delta_g)} + \sqrt{\log(1/\delta_g) + \epsilon_g}\right)^2, \\
\halt, \text{ if } \frac{1}{2}\sum_{t=1}^k \epsilon_t^2 > \left(-\sqrt{\log(1/\delta_g)} + \sqrt{\log(1/\delta_g) + \epsilon_g}\right)^2
\end{cases}$$
is a valid DP filter.
\end{theorem}

\begin{proof}
By conversions between DP and zCDP \cite{bun2016concentrated},  we know that $\epsilon_t$-DP implies $\frac{1}{2}\epsilon_t^2$-zCDP, that is $(\alpha, \frac{1}{2}\epsilon_t^2\alpha)$-RDP, for all $\alpha\geq 1$. Thus, a R\'enyi filter with parameters $(\alpha,\alpha B^\star)$ would stop once $\frac{1}{2}\sum_{t=1}^k\epsilon_t^2 > B^\star$. Since this condition is independent of $\alpha$, the output of adaptive composition with this stopping condition satisfies $(\alpha,B^\star)$-RDP for all $\alpha\geq 1$. This guarantee is equivalent to $B^\star$-zCDP, and by assumption this implies $(\epsilon_g,\delta_g)$-DP as well.

By Fact \ref{fact:conversion}, $B^\star$-zCDP implies $\left(\min_\alpha \alpha B^\star + \frac{\log(1/\delta_g)}{\alpha-1},\delta_g\right)$-DP. Optimizing over $\alpha$ and solving for $B^\star$ such that $\min_\alpha \alpha B^\star + \frac{\log(1/\delta_g)}{\alpha-1} = \epsilon_g$ yields $B^\star = \left(-\sqrt{\log(1/\delta_g)} + \sqrt{\log(1/\delta_g) + \epsilon_g}\right)^2$.
\end{proof}

If the privacy parameters are fixed up front and $\epsilon_t \equiv \epsilon$, simplifying the stopping criterion of the above DP filter implies that adaptive composition of $k$ $\epsilon$-differentially private algorithm satisfies
$$\left(\frac{1}{2}k\epsilon^2 + \sqrt{2k\log(1/\delta)}\epsilon,\delta\right)\text{-differential privacy,}$$
for all $\delta>0$. This tightens the rate of \citet{rogers2016privacy}, whose filter halts when
$$\frac{1}{2}k\epsilon(e^\epsilon-1) + \sqrt{2\left(k\epsilon^2 + C(\epsilon_g,\delta_g)\right)\left(1 + 0.5\log\left(1 + \frac{k\epsilon^2}{C(\epsilon_g,\delta_g)}\right)\right)\log(1/\delta_g)}> \epsilon_g,$$
where $C(\epsilon_g,\delta_g) = \frac{\epsilon_g^2}{28.04\log(1/\delta_g)}$. The factor $C(\epsilon_g,\delta_g)$ essentially determines the gap between our analysis and the analysis of Rogers et al., and our filter is noticeably tighter for non-negligible values of $C(\epsilon_g,\delta_g)$. In addition, our filter has an arguably simpler stopping criterion.

Further improvements on the rate are possible via a more intricate conversion between zCDP and DP, as presented in \cite{bun2016concentrated}.


\subsection{Tracking privacy loss via multiple filters}

A privacy filter is meant to shape the course of adaptive composition by limiting the incurred privacy loss. In practice, one might want to track the privacy loss incurred so far without constraining the analyses. \citet{rogers2016privacy} formalize this desideratum in terms of a \emph{privacy odometer}. For every fixed $t\in\N$, an odometer provides a random bound on the privacy spent so far as a function of the adaptively chosen privacy parameters.
We observe that a sequence of valid R\'enyi privacy filters can be utilized for what is essentially an inverse task: for a fixed sequence of privacy losses $B_1\leq B_2\leq B_3 \dots$, provide random times $T_1 \leq T_2 \leq T_3 \dots$ such that the privacy spent up to time $T_i$ is at most $B_i$. For convenience, we assume that $B_{i+1}-B_i\equiv \Delta$ for a fixed discretization parameter $\Delta>0$ and for all $i$, however in principle the values $\Delta_i = B_{i+1}-B_i$ can vary with $i$. Subsequent work by \citet{lecuyer2021practical} uses a similar discretization idea to construct a privacy odometer in the sense of \citet{rogers2016privacy} based on our R\'enyi privacy filter.


We denote by $\rho_t$ the RDP parameter of order $\alpha$ of $\A_t$ conditional on the past reports, as in equation~\eqref{eqn:generalrho}.

\begin{algorithm}[H]
\SetAlgoLined
\SetKwInOut{Input}{input}
\Input{dataset $S\in\X^n$, discretization error $\Delta>0$, sequence of algorithms $\A_t,t=1,2,\dots,k$}
\hspace{0.1cm} Initialize tracker $O_1 = \Delta$\newline
Set $\Trestart = 1$\newline
\For{$t=1,2,\dots,k$}{
\ Compute $a_t = \A_t(a_1,\dots,a_{t-1},S)$\newline
\uIf{$\F_{\alpha,\Delta}(\rho_{\Trestart},\dots,\rho_{t}) = \halt$}{
\ Augment tracker $O_{t} \leftarrow O_{t-1} + \Delta$\newline
Update restart time $\Trestart \leftarrow t$}
\Else{
\ $O_{t} \leftarrow O_{t-1}$
  }
}
\caption{Tracking privacy loss via multiple privacy filters}
\label{alg:renyi-odometer}
\end{algorithm}

In words, every time an RDP filter with privacy budget $\Delta$ halts, we restart a new filter and augment the tracker by $\Delta$. Here, $\Delta>0$ is the discretization error of the tracker. An important question here is how one should go about choosing $\Delta$. If $\Delta$ is large, then the tracker is very coarse and inaccurate. On the other end, if $\Delta$ is small, the filters might halt often, and whenever a filter halts we effectively make the upper bound on the tracker a bit looser. Roughly speaking, if we restart at time $t$ we lose a factor of $\Delta - \sum_{j = \Trestart}^{t-1} \rho_j$, where $\Trestart$ is the last restart time before $t$.

We state the guarantees of Algorithm \ref{alg:renyi-odometer}. For all $j\in\N$, let $T_j$ denote the step before the $j$-th time a filter restarts in Algorithm \ref{alg:renyi-odometer}. More formally, we can define the sequence $\{T_j\}_j$ recursively as\footnote{We think of the minimum of an empty set as $\infty$.} 
$$T_j = \min\left\{k,\min\left\{t>T_{j-1}:\F_{\alpha,\Delta}(\rho_{T_{j-1}+1},\dots,\rho_{t+1}) = \halt\right\}\right\}, \text{ where } T_0=0.$$

\begin{proposition}
\label{prop:odometer}
Fix $m\in\N$, and suppose that $\rho_j \leq \Delta$ almost surely, for all $j\in \N$. Then, the tracker in Algorithm \ref{alg:renyi-odometer} satisfies $\sup_{(S,S')\in\cS}D^\leftrightarrow_\alpha\left(\A^{(T_m)}(S)\|\A^{(T_m)}(S')\right) \leq m\Delta = O_{T_m}$.
\end{proposition}

\begin{proof}
The algorithm $\A^{(T_m)}$ can be written as an adaptive composition of $m$ algorithms, each of which outputs $(a_{T_{j-1}+1},\dots,a_{T_j})$, $j\in\{1,\dots,m\}$. Therefore, by the standard adaptive composition theorem for RDP, it suffices to argue that each of these $m$ algorithms is RDP, conditional on the outputs of the previous algorithms. Since $\F_{\alpha,\Delta}$ is a valid R\'enyi privacy filter by Theorem \ref{theorem:rdp-at-stopping-times}, each of these $m$ algorithms is indeed $(\alpha,\Delta)$-RDP, which completes the proof.
\end{proof}

Notice that Proposition \ref{prop:odometer} immediately implies that Algorithm \ref{alg:renyi-odometer} is also valid for any $T$ such that $T_{m-1}\leq T\leq T_m$, since $(a_1,\dots,a_T)$ is a post-processing of $(a_1,\dots,a_{T_m})$.


In the context of individual privacy, Proposition \ref{prop:odometer} allows designing a personalized privacy tracker for all analyzed data points. Here, we track $O_t^{(i)}$ for all points $X_i\in S$. The update is analogous to that of Algorithm \ref{alg:renyi-odometer}, the difference being that a separate privacy filter is applied to the individual privacy parameters for all points separately. Naturally, each data point has its own random times of filter exceedances, formally defined as
$$T_j^{(i)} = \min\left\{k,\min\left\{t>T_{j-1}^{(i)}:\F_{\alpha,\Delta}\left(\rho^{(i)}_{T^{(i)}_{j-1}+1},\dots,\rho_{t+1}^{(i)}\right) = \halt\right\}\right\}, \text{ where } T^{(i)}_0=0.$$

Here, $\rho_t^{(i)}$ are individual privacy parameters, measured according to equation \eqref{eqn:stronger-ind-rho}.


It is worth pointing out that the individual values $O_t^{(i)}$ are \emph{sensitive}, as they depend on the value of the data point $X_i$. Importantly, they can be disclosed to the respective user without violating the other users' privacy; $O_t^{(i)}$ depends on $X_i$, but it does not depend on the other data points (other than through $a_1,\dots,a_{t-1}$, which are reported in a privacy-preserving manner).

Below we state an immediate corollary of Proposition \ref{prop:odometer}.

\begin{corollary}
\label{corollary:individual-odometers}
Fix $m\in\N$, and suppose that $\rho_j^{(i)} \leq \Delta$ almost surely, for all $j\in \N$ and $i\in[n]$. Then, the individual tracker $O_t^{(i)}$ upper bounds the individual privacy loss of point $X_i$ at time $T_m^{(i)}$:
$$D^\leftrightarrow_\alpha\left(\A^{(T_m^{(i)})}(S)\|\A^{(T_m^{(i)})}(S^{-i})\right) \leq m\Delta = O_{T_m}^{(i)}.$$
\end{corollary}

Moreover, the same guarantee holds if accounting is done according to the per-instance notion of individual privacy (Definition \ref{def:general-ind-privacy}). In that case, however, the values $O_t^{(i)}$ depend on the whole dataset $S$, and not just $X_i$. Consequently, reporting these values to users, without violating the other users' privacy, would require greater care.

\section{Private gradient descent with individual privacy accounting}
\label{applications}

We discuss an application of the individual privacy filter from Section \ref{filter} to differentially private optimization.

A popular approach to differentially private model training via gradient descent is to clip the norm of individual gradients at every time step and add Gaussian noise to the clipped gradients \cite{abadi2016deep}. Existing analyses compute the overall privacy spent up to a given round by using a uniform upper bound on the gradient norms, determined by the clipping value. Using the individual privacy filter from Section \ref{filter}, we develop a less conservative version of private gradient descent, one which takes into account the \emph{realized} norms of the gradients, rather than just their upper bound.

There are various natural ways one could incorporate individual privacy accounting into the standard private gradient descent (GD) algorithm \cite{abadi2016deep}. To facilitate the comparison, we present a particularly simple one. As in private gradient descent, at every step we clip all computed gradients and add Gaussian noise. However, after the round at which private gradient descent would halt, we additionally look at the ``leftover'' privacy budget for all points, and utilize them until their budget runs out. The leftover budget for each point is essentially equivalent to the difference between the worst-case sum of squared $\ell_2$-norms of the gradients (determined by the clipping value) and the sum of squared $\ell_2$-norms of the \emph{realized} gradients. Below we contrast private gradient descent with individual filtering with the standard private gradient descent algorithm.

\vspace{-0.5cm}
\hspace{-0.55cm}
\begin{minipage}[t]{0.47\textwidth}
\null
\begin{algorithm}[H]
\SetAlgoLined
\SetKwInOut{Input}{input}
\Input{dataset $(X_1,\dots,X_n)$, loss function $\ell(\theta;X_i)$, learning rate $(\eta_t)_{t=1}^\infty$, noise scale $\sigma >0$, clip value $C >0$, number of steps $k\in\N$}
Initialize $\theta_1$ arbitrarily\newline
\For{$t=1,2,\dots,k$}{
\ Compute gradients $g_t(X_i)\leftarrow \nabla_{\theta} \ell(\theta_t;X_i), \forall i$\newline
Clip $\bar g_t(X_i)\leftarrow g_t(X_i)\cdot\min\left(1,\frac{C}{\|g_t(X_i)\|_2}\right), \forall i$\newline
Add noise $\widetilde g_t \leftarrow \frac{1}{n} \sum_{i=1}^n (\bar g_t(X_i) + N(0,\sigma^2C^2\I))$\newline
Take gradient step $\theta_{t+1} \leftarrow \theta_t - \eta_t \widetilde g_t$
}
 Return $\theta_{k+1}$
\caption{Private gradient descent}
\label{alg:private-gd}
\end{algorithm}
\end{minipage}
\hspace{0.25cm}
\begin{minipage}[t]{0.5\textwidth}
\null
\begin{algorithm}[H]
\SetAlgoLined
\SetKwInOut{Input}{input}
\Input{dataset $(X_1,\dots,X_n)$, loss function $\ell(\theta;X_i)$, learning rate $(\eta_t)_{t=1}^\infty$, noise scale $\sigma >0$, clip value $C >0$, number of steps $k_{\max}\in\N$, squared norm budget $\bnorm>0$}
Initialize $\theta_1$ arbitrarily\newline
\For{$t=1,2,\dots,k_{\max}$}{
\ Compute gradients $g_t(X_i)\leftarrow \nabla_{\theta} \ell(\theta_t;X_i), \forall i$\newline
Clip $\bar g_t(X_i)\leftarrow$\newline $~~~~~~g_t(X_i)\cdot\min\left(1,\frac{\min\left(C, \sqrt{\bnorm - \sum_{j=1}^{t-1} \|\bar g_j(X_i)\|_2^2}\right)}{\|g_t(X_i)\|_2}\right), \forall i$\newline
Add noise $\widetilde g_t \leftarrow \frac{1}{n} \sum_{i=1}^n (\bar g_t(X_i) + N(0,\sigma^2C^2\I))$\newline
Take gradient step $\theta_{t+1} \leftarrow \theta_t - \eta_t \widetilde g_t$
}
 Return $\theta_{k+1}$
\caption{Private gradient descent with filtering}
\label{alg:filtered-gd}
\end{algorithm}
\end{minipage}
\vspace{0.3cm}

In Algorithm \ref{alg:private-gd}, all gradients get clipped to have norm at most $C$ at every time step. In Algorithm \ref{alg:filtered-gd}, the gradient for point $X_i$ gets clipped to have norm at most $\min\left(C, \sqrt{\bnorm - \sum_{j=1}^{t-1} \|\bar g_j(X_i)\|_2^2}\right)$. This means that, at least for the first $\lfloor\bnorm/C^2\rfloor$ rounds, all gradients get clipped to have norm at most $C$. After round $\lfloor\bnorm/C^2\rfloor$, points adaptively get filtered out once the accumulated squared norm of their (clipped) gradients reaches $\bnorm$. Therefore, we observe that for $\bnorm = k C^2$ and $k_{\max}=k$, Algorithm \ref{alg:filtered-gd} recovers Algorithm \ref{alg:private-gd}.

The standard privacy guarantees of Algorithm \ref{alg:private-gd} are given as follows.

\begin{proposition}
Private gradient descent (Algorithm \ref{alg:private-gd}) satisfies $\frac{k}{2\sigma^2}$-zCDP, or, equivalently, $\left(\alpha, \frac{\alpha k}{2\sigma^2}\right)$-RDP for all $\alpha\geq 1$.
\end{proposition}

We prove the privacy guarantees of Algorithm \ref{alg:filtered-gd} as a corollary of our individual privacy filter.

\begin{proposition}
Private gradient descent with filtering (Algorithm \ref{alg:filtered-gd}) satisfies $\frac{\bnorm}{2\sigma^2C^2}$-zCDP, or, equivalently, $\left(\alpha,\frac{\alpha\bnorm}{2\sigma^2C^2}\right)$-RDP for all $\alpha\geq 1$.
\end{proposition}

\begin{proof}
By properties of the Gaussian mechanism, the individual RDP parameters of order $\alpha$ are $\rho_t^{(i)} = \frac{\alpha \|\bar g_t(X_i)\|_2^2}{2\sigma^2}$. Therefore, by properties of the individual filter, as long as $ \frac{\alpha \sum_{j=1}^t\|\bar g_j(X_i)\|_2^2}{2\sigma^2}\leq  B$, the output is $(\alpha,B)$-individually RDP for $X_i$. The clipping step ensures this inequality holds with $B=\frac{\alpha\bnorm}{2\sigma^2}$ for all $t\in\N$ and for all data points $X_i$, and therefore the algorithm is $\left(\alpha,\frac{\alpha\bnorm}{2\sigma^2}\right)$-RDP.
\end{proof}

When $\bnorm = k C^2$, the privacy guarantees of Algorithm \ref{alg:filtered-gd} are the same as those of Algorithm \ref{alg:private-gd}. However, they \emph{do not} depend on the total number of steps $k_{\max}$---in particular, $k_{\max}$ need not be equal to $k$, which raises the question of how to set $k_{\max}$. (Certainly $k_{\max}$ should be at least $\lfloor\bnorm/C^2\rfloor$, otherwise the privacy budget is not used up for any data point.) If $k_{\max}$ is relatively small, we might stop the optimization process too early, and thus forgo a potentially higher accuracy. If $k_{\max}$ is too large, then a lot of points might get filtered out and we might add high amounts of noise relative to the number of active points.

One solution is to periodically estimate the number of active points in a privacy-preserving fashion. In particular, after round $\lfloor\bnorm/C^2\rfloor$, the analyst can estimate the size of the active set $\{ i\ :\ \sum_{j=1}^t\|\bar  g_j(X_i)\|_2^2\leq \bnorm\}$ (which is just a linear query) and use it to stop. To reduce the privacy cost of such estimates one can use the continual monitoring technique \citep{dwork2010differential} since each point is filtered out only once. Alternatively, if one only wants to ensure that the size of the active set exceeds some fixed threshold, one can use the sparse vector technique \cite{dwork2009complexity, dwork2014algorithmic} and thus incur an even smaller privacy loss due to adaptive stopping. Another solution is to periodically check the training accuracy, which is again a linear query, and stop once it plateaus.

\subsection{Experiments}

As proof of concept, we compare the performance of private gradient descent (Algorithm \ref{alg:private-gd}) and its generalization with filtering (Algorithm \ref{alg:filtered-gd}) by training a convolutional neural network on MNIST \cite{lecun2010mnist}. We use the default architecture from the MNIST example of the PyTorch Opacus library \cite{Opacus}.

We remark that, in practice, it is more common to use private SGD, rather than batch GD. While in principle it is possible to compute individual privacy parameters for SGD, random subsampling of points requires computing gradients for \emph{all} points at every step to observe gains from individual accounting. 
As a result, SGD is no less computationally expensive than batch GD in the context of individual accounting. Nevertheless, GD requires fewer steps and---importantly--- we observe that it achieves a significantly better privacy-utility tradeoff. For example, using the same architecture, the Opacus library reports accuracy $(94.63\% \pm 0.34\%)$ for $\epsilon = 1.16$ and $\delta = 10^{-5}$, which is almost the same accuracy we obtain with $\epsilon = 0.5$ and $\delta = 10^{-5}$ (see table below). Recent large-scale experiments on differentially private training of language models \cite{anil2021large} similarly point to larger batches leading to higher utility, and we believe this phenomenon likely holds in many other settings and is worth further exploration.

All reported average accuracies and deviations are estimated over 10 trials. We fix the target differential privacy parameters $(\epsilon,\delta)$, and evaluate the test accuracy. We set $\delta=10^{-5}$, and vary the value of $\epsilon$. In the first set of evaluations, for every $\epsilon$ we tune all algorithm hyperparameters to achieve high test accuracy with private gradient descent. For private gradient descent with filtering, to make the comparison as clear as possible, we adopt the same hyperparameters. After $\bnorm/C^2$ steps, we query the training accuracy a fixed number of times in intervals of 5 steps and stop adaptively, when the training accuracy is observed to be the highest. We provide the specifics of the stopping rule and other experimental details in the Appendix.

\begin{table}[h]
\centering
\begin{tabular}{ccc}
\hline
$\epsilon$ & private GD (tuned) & private GD (tuned) with filtering \\ \hline
0.3     & $(93.29 \pm 0.49)\%$                     & $(93.64 \pm 0.46)\%$                            \\ \hline
0.5     & $(94.62\pm 0.43)\%$                    & $(94.90\pm 0.26)\%$                              \\ \hline
1.0       & $(96.25\pm0.23)\%$                     & $(96.25\pm0.23)\%$                              
\end{tabular}
\end{table}

\begin{figure}[t]
    \centering
    \includegraphics[width = 0.4\textwidth]{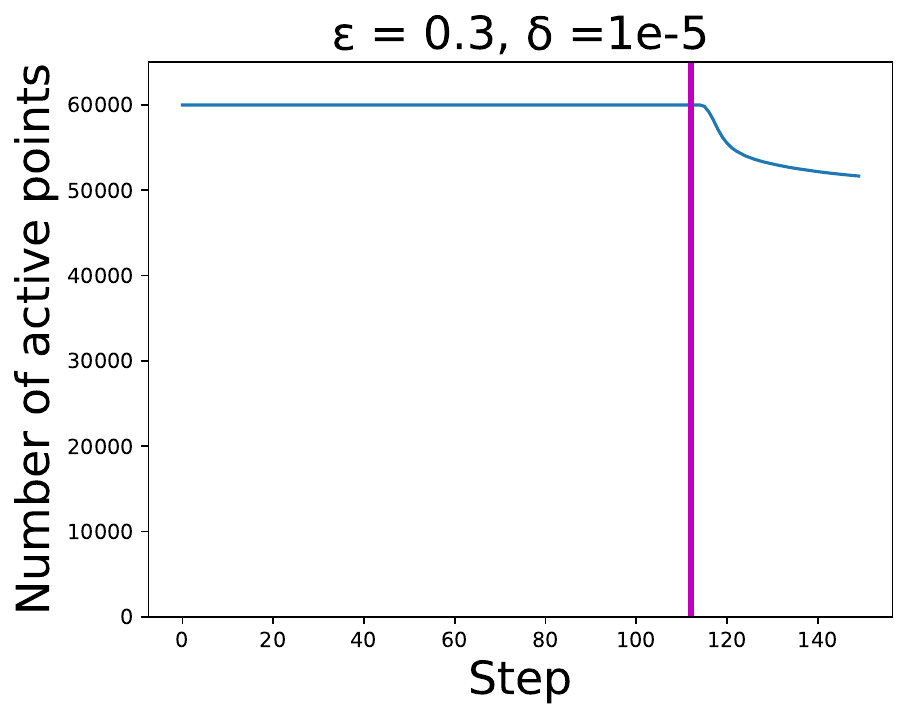}
    \includegraphics[width = 0.4\textwidth]{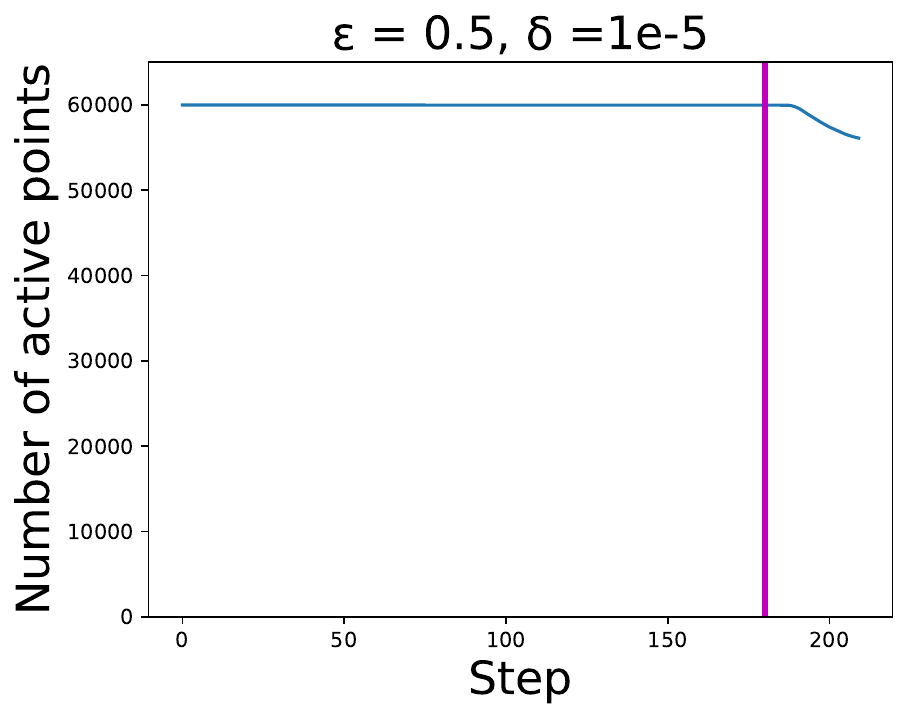}
    \caption{Number of active points during one run of private GD with filtering in the tuned regime, for $\epsilon = 0.3$ (left) and $\epsilon=0.5$ (right). The solid vertical line denotes step $\bnorm/C^2$.}
    \label{fig:active_points}
\end{figure}

Overall we observe modest accuracy improvements with individual filtering in the tuned regime. The benefits are more noticeable for small $\epsilon$, while for large $\epsilon$, the accuracy plateaus after $\bnorm/C^2$ steps and hence we do not add extra steps. In Figure \ref{fig:active_points} we plot the number of active points, i.e. those that have not yet exhausted their privacy budget, for $\epsilon \in\{ 0.3, 0.5\}$. Due to extensive hyperparameter tuning in this specific application, private GD is implicitly tuned to clip the gradients in such a way that hard-to-classify points exhaust their privacy budget. Such tuning, however, is not possible when queries are chosen by human analysts (which is also hard to run experiments on) and in federated settings where data is held by the clients and finding the optimal setting of hyperparameters is typically infeasible. 
Therefore we examine the benefits of individual filtering in examples of such suboptimally tuned settings. For example, we evaluate the performance when the clipping value $C$ is chosen to be larger than in the optimal setting (while keeping the noise level the same). Specifically, for $\epsilon\in\{0.3,0.5\}$, we set $C$ to be $1.5$ the optimally tuned value, and for $\epsilon=1.0$ we set $C$ to be double the tuned value. We reduce the number of optimization steps accordingly to achieve the same privacy guarantee. In the suboptimal regime, the benefits of individual filtering become much more significant as a large fraction of points remain in the active pool after $\bnorm/C^2$ steps.

\begin{table}[h]
\centering
\begin{tabular}{ccc}
\hline
$\epsilon$ & private GD (suboptimal clipping) & private GD (suboptimal clipping) with filtering \\ \hline
0.3  &   $(84.47 \pm 3.95)\%$                      & $(92.25 \pm 0.91)\%$                            \\ \hline
0.5     & $(92.07\pm2.07)\%$                     & $(94.30\pm 0.58)\%$                               \\ \hline
1.0       & $(94.45\pm 0.45)\%$                    & $(95.33\pm 0.34)\%$ 
\end{tabular}
\end{table}

Similar results are obtained when the noise rate is not set optimally. Here, we decrease $\sigma$ by a factor of $1.5$ for $\epsilon\in\{0.3,0.5\}$ and by a factor of $2$ for $\epsilon = 1.0$. We adjust the number of optimization steps accordingly, and keep all other hyperparameters as in the tuned regime. As before, we observe benefits to performing additional steps, especially for small values of $\epsilon$.

\begin{table}[h]
\centering
\begin{tabular}{ccc}
\hline
$\epsilon$ & private GD (suboptimal noise level) & private GD (suboptimal noise level) with filtering \\ \hline
0.3 & $(86.88 \pm 2.28)\%$                    & $(91.20 \pm 0.73)\%$                            \\ \hline
0.5 & $(92.37\pm1.32)\%$                    & $(93.86\pm 0.39)\%$                              \\ \hline
1.0 & $(94.35\pm 0.23)\%$                     & $(94.50\pm 0.14)\%$  
\end{tabular}

\end{table}

Altogether, we view this application as a useful proof of concept: it demonstrates that individual accounting is practical, easy to implement, and can only make the results better.

\subsection*{Acknowledgements}
We thank Katrina Ligett, Kunal Talwar, and Neil Vexler for insightful discussions on individual notions of privacy and feedback on this work. We are grateful to Ryan McKenna for providing code and suggestions for the experiments.

\bibliography{ind_privacy}

\begin{thebibliography}{34}
\providecommand{\natexlab}[1]{#1}
\providecommand{\url}[1]{\texttt{#1}}
\expandafter\ifx\csname urlstyle\endcsname\relax
  \providecommand{\doi}[1]{doi: #1}\else
  \providecommand{\doi}{doi: \begingroup \urlstyle{rm}\Url}\fi

\bibitem[Abadi et~al.(2016)Abadi, Chu, Goodfellow, McMahan, Mironov, Talwar,
  and Zhang]{abadi2016deep}
Martin Abadi, Andy Chu, Ian Goodfellow, H~Brendan McMahan, Ilya Mironov, Kunal
  Talwar, and Li~Zhang.
\newblock Deep learning with differential privacy.
\newblock In \emph{Proceedings of the 2016 ACM SIGSAC Conference on Computer
  and Communications Security}, pages 308--318, 2016.

\bibitem[Anil et~al.(2021)Anil, Ghazi, Gupta, Kumar, and
  Manurangsi]{anil2021large}
Rohan Anil, Badih Ghazi, Vineet Gupta, Ravi Kumar, and Pasin Manurangsi.
\newblock Large-scale differentially private {BERT}.
\newblock \emph{arXiv preprint arXiv:2108.01624}, 2021.

\bibitem[Bassily et~al.(2016)Bassily, Nissim, Smith, Steinke, Stemmer, and
  Ullman]{bassily2016algorithmic}
Raef Bassily, Kobbi Nissim, Adam Smith, Thomas Steinke, Uri Stemmer, and
  Jonathan Ullman.
\newblock Algorithmic stability for adaptive data analysis.
\newblock In \emph{Proceedings of the forty-eighth annual ACM symposium on
  Theory of Computing}, pages 1046--1059, 2016.

\bibitem[Bun and Steinke(2016)]{bun2016concentrated}
Mark Bun and Thomas Steinke.
\newblock Concentrated differential privacy: Simplifications, extensions, and
  lower bounds.
\newblock In \emph{Theory of Cryptography Conference}, pages 635--658.
  Springer, 2016.

\bibitem[Cesar and Rogers(2021)]{cesar2020unifying}
Mark Cesar and Ryan Rogers.
\newblock Bounding, concentrating, and truncating: Unifying privacy loss
  composition for data analytics.
\newblock In \emph{Algorithmic Learning Theory}, pages 421--457, 2021.

\bibitem[Cummings and Durfee(2020)]{cummings2020individual}
Rachel Cummings and David Durfee.
\newblock Individual sensitivity preprocessing for data privacy.
\newblock In \emph{Proceedings of the Fourteenth Annual ACM-SIAM Symposium on
  Discrete Algorithms}, pages 528--547. SIAM, 2020.

\bibitem[Cummings et~al.(2018)Cummings, Krehbiel, Lai, and
  Tantipongpipat]{cummings2018differential}
Rachel Cummings, Sara Krehbiel, Kevin~A Lai, and Uthaipon Tantipongpipat.
\newblock Differential privacy for growing databases.
\newblock In \emph{Advances in Neural Information Processing Systems}, pages
  8864--8873, 2018.

\bibitem[Dong et~al.(2019)Dong, Roth, and Su]{dong2019gaussian}
Jinshuo Dong, Aaron Roth, and Weijie~J Su.
\newblock Gaussian differential privacy.
\newblock \emph{arXiv preprint arXiv:1905.02383}, 2019.

\bibitem[Dwork and Lei(2009)]{dwork2009differential}
Cynthia Dwork and Jing Lei.
\newblock Differential privacy and robust statistics.
\newblock In \emph{Proceedings of the forty-first annual ACM symposium on
  Theory of computing}, pages 371--380, 2009.

\bibitem[Dwork and Roth(2014)]{dwork2014algorithmic}
Cynthia Dwork and Aaron Roth.
\newblock The algorithmic foundations of differential privacy.
\newblock \emph{Foundations and Trends in Theoretical Computer Science},
  9\penalty0 (3-4):\penalty0 211--407, 2014.

\bibitem[Dwork and Rothblum(2016)]{dwork2016concentrated}
Cynthia Dwork and Guy~N Rothblum.
\newblock Concentrated differential privacy.
\newblock \emph{arXiv preprint arXiv:1603.01887}, 2016.

\bibitem[Dwork et~al.(2006{\natexlab{a}})Dwork, Kenthapadi, McSherry, Mironov,
  and Naor]{ODO06}
Cynthia Dwork, Krishnaram Kenthapadi, Frank McSherry, Ilya Mironov, and Moni
  Naor.
\newblock Our data, ourselves: Privacy via distributed noise generation.
\newblock In \emph{Advances in Cryptology---EUROCRYPT}, pages 486--503,
  2006{\natexlab{a}}.

\bibitem[Dwork et~al.(2006{\natexlab{b}})Dwork, McSherry, Nissim, and
  Smith]{DMNS06}
Cynthia Dwork, Frank McSherry, Kobbi Nissim, and Adam Smith.
\newblock Calibrating noise to sensitivity in private data analysis.
\newblock In \emph{Proc. of the Third Conf. on Theory of Cryptography (TCC)},
  pages 265--284, 2006{\natexlab{b}}.
\newblock URL \url{http://dx.doi.org/10.1007/11681878_14}.

\bibitem[Dwork et~al.(2009)Dwork, Naor, Reingold, Rothblum, and
  Vadhan]{dwork2009complexity}
Cynthia Dwork, Moni Naor, Omer Reingold, Guy~N Rothblum, and Salil Vadhan.
\newblock On the complexity of differentially private data release: efficient
  algorithms and hardness results.
\newblock In \emph{Proceedings of the forty-first annual ACM symposium on
  Theory of computing}, pages 381--390, 2009.

\bibitem[Dwork et~al.(2010{\natexlab{a}})Dwork, Naor, Pitassi, and
  Rothblum]{dwork2010differential}
Cynthia Dwork, Moni Naor, Toniann Pitassi, and Guy~N Rothblum.
\newblock Differential privacy under continual observation.
\newblock In \emph{Proceedings of the forty-second ACM symposium on Theory of
  computing}, pages 715--724, 2010{\natexlab{a}}.

\bibitem[Dwork et~al.(2010{\natexlab{b}})Dwork, Rothblum, and
  Vadhan]{dwork2010boosting}
Cynthia Dwork, Guy~N Rothblum, and Salil Vadhan.
\newblock Boosting and differential privacy.
\newblock In \emph{2010 IEEE 51st Annual Symposium on Foundations of Computer
  Science}, pages 51--60. IEEE, 2010{\natexlab{b}}.

\bibitem[Dwork et~al.(2015{\natexlab{a}})Dwork, Feldman, Hardt, Pitassi,
  Reingold, and Roth]{dwork2015generalization}
Cynthia Dwork, Vitaly Feldman, Moritz Hardt, Toni Pitassi, Omer Reingold, and
  Aaron Roth.
\newblock Generalization in adaptive data analysis and holdout reuse.
\newblock In \emph{Advances in Neural Information Processing Systems}, pages
  2350--2358, 2015{\natexlab{a}}.

\bibitem[Dwork et~al.(2015{\natexlab{b}})Dwork, Feldman, Hardt, Pitassi,
  Reingold, and Roth]{dwork2015preserving}
Cynthia Dwork, Vitaly Feldman, Moritz Hardt, Toniann Pitassi, Omer Reingold,
  and Aaron~Leon Roth.
\newblock Preserving statistical validity in adaptive data analysis.
\newblock In \emph{Proceedings of the forty-seventh annual ACM symposium on
  Theory of computing}, pages 117--126, 2015{\natexlab{b}}.

\bibitem[Ebadi et~al.(2015)Ebadi, Sands, and Schneider]{ebadi2015differential}
Hamid Ebadi, David Sands, and Gerardo Schneider.
\newblock Differential privacy: Now it's getting personal.
\newblock \emph{Acm Sigplan Notices}, 50\penalty0 (1):\penalty0 69--81, 2015.

\bibitem[Erlingsson et~al.(2019)Erlingsson, Feldman, Mironov, Raghunathan,
  Talwar, and Thakurta]{erlingsson2019amplification}
{\'U}lfar Erlingsson, Vitaly Feldman, Ilya Mironov, Ananth Raghunathan, Kunal
  Talwar, and Abhradeep Thakurta.
\newblock Amplification by shuffling: From local to central differential
  privacy via anonymity.
\newblock In \emph{Proceedings of the Thirtieth Annual ACM-SIAM Symposium on
  Discrete Algorithms}, pages 2468--2479. SIAM, 2019.

\bibitem[Feldman and Steinke(2018)]{feldman2018calibrating}
Vitaly Feldman and Thomas Steinke.
\newblock Calibrating noise to variance in adaptive data analysis.
\newblock In \emph{Conference On Learning Theory}, pages 535--544, 2018.

\bibitem[Ghosh and Roth(2011)]{ghosh2011selling}
Arpita Ghosh and Aaron Roth.
\newblock Selling privacy at auction.
\newblock In \emph{Proceedings of the 12th ACM conference on Electronic
  commerce}, pages 199--208, 2011.

\bibitem[Kairouz et~al.(2017)Kairouz, Oh, and
  Viswanath]{kairouz2017composition}
Peter Kairouz, Sewoong Oh, and Pramod Viswanath.
\newblock The composition theorem for differential privacy.
\newblock \emph{IEEE Transactions on Information Theory}, 63\penalty0
  (6):\penalty0 4037--4049, 2017.

\bibitem[LeCun et~al.(2010)LeCun, Cortes, and Burges]{lecun2010mnist}
Yann LeCun, Corinna Cortes, and CJ~Burges.
\newblock {MNIST} handwritten digit database.
\newblock \emph{ATT Labs [Online]. Available:
  http://yann.lecun.com/exdb/mnist}, 2, 2010.

\bibitem[L{\'e}cuyer(2021)]{lecuyer2021practical}
Mathias L{\'e}cuyer.
\newblock Practical privacy filters and odometers with {R}\'enyi differential
  privacy and applications to differentially private deep learning.
\newblock \emph{arXiv preprint arXiv:2103.01379}, 2021.

\bibitem[Ligett et~al.(2020)Ligett, Peale, and Reingold]{ligett2020bounded}
Katrina Ligett, Charlotte Peale, and Omer Reingold.
\newblock Bounded-leakage differential privacy.
\newblock In \emph{1st Symposium on Foundations of Responsible Computing
  (FORC)}, 2020.

\bibitem[Mironov(2017)]{mironov2017renyi}
Ilya Mironov.
\newblock R{\'e}nyi differential privacy.
\newblock In \emph{IEEE 30th Computer Security Foundations Symposium (CSF)},
  pages 263--275. IEEE, 2017.

\bibitem[Murtagh and Vadhan(2016)]{murtagh2016complexity}
Jack Murtagh and Salil Vadhan.
\newblock The complexity of computing the optimal composition of differential
  privacy.
\newblock In \emph{Theory of Cryptography Conference}, pages 157--175.
  Springer, 2016.

\bibitem[Nissim et~al.(2007)Nissim, Raskhodnikova, and Smith]{nissim2007smooth}
Kobbi Nissim, Sofya Raskhodnikova, and Adam Smith.
\newblock Smooth sensitivity and sampling in private data analysis.
\newblock In \emph{Proceedings of the thirty-ninth annual ACM symposium on
  Theory of computing}, pages 75--84, 2007.

\bibitem[Rogers et~al.(2016)Rogers, Roth, Ullman, and
  Vadhan]{rogers2016privacy}
Ryan~M Rogers, Aaron Roth, Jonathan Ullman, and Salil Vadhan.
\newblock Privacy odometers and filters: Pay-as-you-go composition.
\newblock \emph{Advances in Neural Information Processing Systems},
  29:\penalty0 1921--1929, 2016.

\bibitem[Sommer et~al.(2019)Sommer, Meiser, and Mohammadi]{sommer2019privacy}
David~M Sommer, Sebastian Meiser, and Esfandiar Mohammadi.
\newblock Privacy loss classes: The central limit theorem in differential
  privacy.
\newblock \emph{Proceedings on privacy enhancing technologies}, 2019\penalty0
  (2):\penalty0 245--269, 2019.

\bibitem[Wang(2019)]{wang2019per}
Yu-Xiang Wang.
\newblock Per-instance differential privacy.
\newblock \emph{Journal of Privacy and Confidentiality}, 9\penalty0 (1), 2019.

\bibitem[Wu et~al.(2019)Wu, Roth, Ligett, Waggoner, and Neel]{wu2019accuracy}
Steven Wu, Aaron Roth, Katrina Ligett, Bo~Waggoner, and Seth Neel.
\newblock Accuracy first: Selecting a differential privacy level for
  accuracy-constrained {ERM}.
\newblock \emph{Journal of Privacy and Confidentiality}, 9\penalty0 (2), 2019.

\bibitem[Yousefpour et~al.(2021)Yousefpour, Shilov, Sablayrolles, Testuggine,
  Prasad, Malek, Nguyen, Ghosh, Bharadwaj, Zhao, Cormode, and Mironov]{Opacus}
A.~Yousefpour, I.~Shilov, A.~Sablayrolles, D.~Testuggine, K.~Prasad, M.~Malek,
  J.~Nguyen, S.~Ghosh, A.~Bharadwaj, J.~Zhao, G.~Cormode, and I.~Mironov.
\newblock Opacus: User-friendly differential privacy library in {PyTorch}.
\newblock \emph{arXiv preprint arXiv:2109.12298}, 2021.

\end{thebibliography}
\bibliographystyle{plainnat}

\appendix

\section{Experimental details}
\label{exp_details}

We train a convolutional neural network using the implementation of private gradient descent from the Opacus PyTorch library \cite{Opacus}. We use the same architecture as in the MNIST example of the library. Since we run batch gradient descent and not SGD as in the library example, we tune all hyperparameters from scratch.

For $\epsilon = 0.3$, we set $\sigma = 170$, $C=10$, $\eta_t \equiv \eta = 0.2$, and $k=112$ for private GD without filtering. To achieve the same privacy guarantees using private GD with individual filtering, we set $\bnorm = k C^2 = 11200$.

For $\epsilon = 0.5$, we set $\sigma = 130$, $C=15$, $\eta_t \equiv \eta = 0.15$, and $k=180$ for private GD without filtering. For GD with individual filtering, we set $\bnorm = k C^2 = 40500$.

For $\epsilon = 1.0$, we set $\sigma = 100$, $C=10$, $\eta_t \equiv \eta = 0.2$, and $k=420$ for private GD without filtering. This parameter configuration achieves accuracy of $(96.25\pm0.23)\%$. When private GD achieves such high accuracies, we observe little benefit to individual filtering. This is due to the fact that the proportion of points filtered out right after round $\lfloor\bnorm/C^2\rfloor$ is comparable to the proportion of points yet misclassified, suggesting that few misclassified points remain in the active pool. Therefore, we set $\bnorm = k C^2 = 42000$, and $k_{\max}=k$. 

To apply the individual filter, after $kC^2$ steps we continue running gradient descent while adaptively dropping points when their privacy budget is exhausted. In particular, starting with step $\lfloor kC^2\rfloor$, we query the training accuracy 8 times in intervals of 5 steps (hence, the total number of additional steps is 35). As the final model we take the iterate when the queried training accuracy is highest.

Note that, technically, these additional queries should be reported in a privacy-preserving manner to formally ensure DP. However, these are simple linear queries that can be reported with high accuracy at little additional privacy cost. For $\epsilon\in\{0.3,0.5\}$, it suffices to report the training accuracy with $1\%$ resolution. Eight such reports require a smaller privacy cost than, say, one or two additional optimization steps. For $\epsilon = 1.0$ it suffices to report the accuracy with $0.1\%$ resolution since the accuracy improvements after $k C^2$ steps are generally smaller for large $\epsilon$. The additional reports for $\epsilon = 1.0$ have the cost of a few dozen extra steps. In either case, the privacy cost of the additional reports is less than $1\%$ of the intended privacy parameter.

In the suboptimal regime, we increase $C$ by a factor of $1.5$ for $\epsilon\in\{0.3,0.5\}$ and accordingly decrease $\sigma$ by the same factor. The number of steps is similarly decreased by a factor of $1.5^2$. We do a similar adjustment for $\epsilon=1.0$, where we increase $C$ by a factor of $2$. Similarly, when we decrease $\sigma$ by a factor of $1.5$ for $\epsilon\in\{0.3,0.5\}$ and by a factor of $2$ for $\epsilon = 1.0$, we adjust the number of optimization steps accordingly, and keep all other hyperparameters as in the tuned regime.

\end{document}